\documentclass[twocolumn,aps,prc,superscriptaddress,amsmath,amssymb,nofootinbib,floatfix]{revtex4-2}

\usepackage{graphicx}
\usepackage{bm}
\usepackage{amsmath,amssymb}
\usepackage{amsthm}
\newenvironment{myproof}{\par\noindent\textit{Proof.}\ }{\hfill$\square$\par}
\usepackage{booktabs}
\usepackage[colorlinks=true,citecolor=blue,linkcolor=blue,urlcolor=blue]{hyperref}

\newcommand{\su}[1]{\mathrm{SU(#1)}}
\newcommand{\uu}[1]{\mathrm{U(#1)}}
\newcommand{\oo}[1]{\mathrm{O(#1)}}
\newcommand{\C}{\hat{C}_2}
\newcommand{\Cthree}{\hat{C}_3}
\newcommand{\Var}{\mathrm{Var}}
\newcommand{\Tr}{\mathrm{Tr}}
\newcommand{\ket}[1]{|#1\rangle}
\newcommand{\bra}[1]{\langle#1|}
\newcommand{\avg}[1]{\langle#1\rangle}
\newcommand{\nd}{\hat{n}_d}
\newcommand{\Pdag}{P^{\dagger}}

\newtheorem{proposition}{Proposition}

\begin{document}

\title{State-resolved quantum-information fingerprints of partial dynamical symmetry\\
in the interacting boson model}

\author{Dhritimalya Roy}
\email{rdhritimalya@gmail.com}
\affiliation{Department of Mathematics, Jadavpur University, Kolkata 700032, India}
\affiliation{Department of Physics, Presidency University, Kolkata 700073, India}

\date{\today}

\begin{abstract}
Partial dynamical symmetry (PDS) is an algebraic structure in which a symmetry is neither exact nor
fully broken: a subset of eigenstates keeps good quantum numbers and stays solvable while the rest
mix. It is currently identified from spectroscopy. We ask whether it also leaves a purely structural
mark in the eigenstates, and find, for the $sd$-IBM Hamiltonians studied here, that it does, though not
in the magnitude of entanglement. The diagnostic is the variance of a symmetry Casimir, the label
variance $\Var\,\C[G]$, which shares its zero with a block-coherence entropy and a block impurity: all
three vanish when a state occupies a single Casimir sector, a single irreducible representation once
conjugate representations are resolved by the cubic Casimir. Resolved state by state, the variance is
numerically zero on the symmetry-selected subset and large on the mixed states, isolating the PDS
states from the wavefunctions alone; the separation does not close as the boson number grows and
persists at Leviatan's first- and second-order critical points, whereas the magnitude of bipartite
$s|d$ entanglement does not separate the two families. We anchor the analysis in $^{168}$Er,
reproducing Leviatan's fitted $\su3$ admixtures and his Hamiltonian-parameter-independent
$\gamma\!\to\! g$ $B(E2)$ ratios. The diagnostic also has an exact dynamical form, a Casimir Loschmidt
echo rigid on the symmetry-pure states, and we demonstrate its readout on an IBM superconducting
processor at $N=3$, where a single solvable/mixed exemplar reproduces the predicted pattern within
finite-shot uncertainties.
\end{abstract}

\maketitle

\section{Introduction}
\label{sec:intro}

The interacting boson model (IBM) describes the low-lying collective states of even--even nuclei
in terms of $s$ ($L=0$) and $d$ ($L=2$) bosons and possesses a rich algebraic
structure~\cite{IachelloArima,VanIsacker2018}. Its three dynamical-symmetry (DS) limits,
$\uu5$, $\su3$ and $\oo6$, correspond to spherical vibrator, axially deformed rotor and
$\gamma$-unstable shapes, and are distinguished by the fact that the entire spectrum is
analytically solvable and every eigenstate carries the good quantum numbers of a nested chain of
algebras. Exact dynamical symmetries are, however, rare in real nuclei. What is ubiquitous, as
Van Isacker's review emphasizes, are their partial and quasi extensions~\cite{VanIsacker2018}.

Partial dynamical symmetry (PDS), introduced by Leviatan, relaxes the requirement that
\emph{all} states be solvable~\cite{Leviatan1996,AlhassidLeviatan1992,Leviatan2011}. In a PDS
Hamiltonian, a structurally selected subset of eigenstates keeps exact quantum numbers and
closed-form energies, while the rest mix. This selectivity survives at the critical
point of a quantum phase transition (QPT), where the competing symmetries are conserved exactly
by some states and broken by others~\cite{Leviatan2007}. Quasi-dynamical symmetry (QDS), introduced
by Rochford and Rowe~\cite{RochfordRowe1988,RoweRochfordRepka1988}, is a
distinct, related notion in which a band of states shares a common intrinsic structure and
mimics a dynamical symmetry over a range of parameters through a coherent mixing of
irreducible representations (irreps). In the interacting boson model, Kremer \emph{et al.}~\cite{Kremer2014}
identified a region exhibiting both structures in the ground band and, notably for our purposes,
labelled the $\oo6$-PDS aspect ``purity'' and the $\su3$-QDS aspect ``coherence,'' the algebraic
vocabulary that we turn into literal density-matrix quantities below, without computing the
corresponding quantum-information objects themselves.

Quantum-information (QI) diagnostics have recently found use in nuclear
structure. Entanglement entropy and mutual information have served as order parameters for
shape-phase transitions in the IBM~\cite{Jafarizadeh2022,Ghapanvari2025}, as probes of pairing
and deformation in the shell model~\cite{PerezObiol2023}, and as signatures of configuration
mixing~\cite{Shinde2026}; nonstabilizerness (``magic'') has been advanced as a resource that
bipartite entropy misses~\cite{Brokemeier2025}. These studies share a common design: they track
the entanglement of the ground state (or a few low-lying states) as an external control parameter
is swept, using entanglement as a mean-field order parameter for a phase transition. None asks
whether individual eigenstates \emph{within a fixed spectrum} differ, and none connects QI to the
symmetry selectivity of PDS.

That gap is the subject of this paper. Existing interacting-boson entanglement studies use the
entanglement of the ground state as a mean-field order parameter for the shape-phase transition
(we contrast them in detail in Sec.~\ref{sec:prior}); here we instead resolve quantum-information
measures state by state across the exact spectrum and along symmetry-preserving parameter flows,
and show that in the $sd$-IBM a fingerprint of the studied partial dynamical symmetry is not the magnitude of
entanglement but the vanishing of a symmetry-label fluctuation on a structurally selected subset of
states.

Our central object is a single measurable with two companions that share its zero condition
(Sec.~\ref{sec:measurable}): the label variance $\Var\,\C[G]$ of an eigenstate for a chosen
subalgebra $G$, the block-coherence entropy $S_G$, and the block impurity $1-\mathcal{P}_G$. All
three vanish precisely when the state carries an exact $G$-irrep label. With this diagnostic we
report three findings. First, the label variance separates the symmetry-pure states from the mixed
ones sharply, numerically zero versus large and well separated, whereas the magnitude of bipartite entanglement does not: at a
stable $\su3$-PDS point the static entanglement of the two families overlaps, and the solvable
states are distinguished instead by the rigidity of their entanglement along the PDS-preserving
directions of the Hamiltonian. Second, the diagnostic survives at criticality (which does not by
itself produce PDS), taking two forms set by the order of the transition: a symmetry-pure subset
amid strong sector mixing at a first-order critical point, and a single conserved label shared by the
whole spectrum at a second-order one. Third, the block purity makes the ``purity'' and
``coherence'' vocabulary of the quasi-dynamical literature literal, and the label variance tracks
neither the multipartite entanglement nor the magic of the states in the cases we examined, so it
is not a restatement of either. We anchor the results in $^{168}$Er, where the block probabilities
match Leviatan's fitted admixtures and reproduce his Hamiltonian-parameter-independent $\gamma\!\to\! g$ $B(E2)$
ratios. The contribution is not the elementary identity that a good symmetry label carries zero
Casimir variance; it is the systematic use of that variance as a state-resolved classifier across an
entire PDS spectrum, together with its separation from closed-form solvability and from entanglement
magnitude, its scaling and controlled-breaking behavior, its two critical-point fingerprints, and its
measurability on a quantum processor.

Throughout, the physics result, a structural, state-resolved criterion for PDS, is the
contribution. The results of Secs.~\ref{sec:stable}--\ref{sec:prior} are obtained by exact
diagonalization; Sec.~\ref{sec:qsim} then encodes the model on a qubit register and shows the same
fingerprint can be prepared and read out variationally, with time evolution appearing only as a
tool and no claim of quantum advantage.

\section{Model and conventions}
\label{sec:model}

We work in the $sd$-IBM at fixed total boson number $N$, the linear Casimir of $\uu6$, in the
$m$-scheme boson Fock basis. All spectra are computed by exact diagonalization in the $M=0$
block, which contains exactly one member of every angular-momentum multiplet. The $\su3$-PDS
Hamiltonian of Leviatan~\cite{Leviatan1996,Leviatan2020} is
\begin{equation}
\hat H = h_0\,\Pdag_0 \tilde P_0 + h_2\,\Pdag_2 \cdot \tilde P_2 + C\,\hat L^2 ,
\label{eq:H}
\end{equation}
with
\begin{align}
\Pdag_0 &= d^{\dagger}\!\cdot d^{\dagger} - 2 (s^{\dagger})^2 , \\
\Pdag_{2\mu} &= 2\, s^{\dagger} d^{\dagger}_{\mu} + \sqrt{7}\,(d^{\dagger} d^{\dagger})^{(2)}_{\mu} .
\end{align}
The operators $\Pdag_0$ and $\Pdag_{2\mu}$ annihilate the $\su3$ intrinsic condensate, so that
$\hat H$ has an analytically solvable set of eigenstates: the ground band $(2N,0)$ with $K=0$,
$L=0,2,\dots,2N$ and energy $E = C\,L(L+1)$, and the $\gamma^k$ bands $(2N-4k,2k)$ with $K=2k$,
$L=2k,\dots,2N-2k$ and energy $E = 6 h_2 k\,(2N-2k+1) + C\,L(L+1)$. The remaining eigenstates are
$\su3$-mixed.

We fix operator normalizations by the identity~\cite{Leviatan2020}
\begin{equation}
\Pdag_0 \tilde P_0 + \Pdag_2 \cdot \tilde P_2 = -\,\C[\su3] + 2\hat N(2\hat N + 3),
\label{eq:arbiter}
\end{equation}
which our implementation reproduces to $\le 2\times10^{-13}$ ($N=3,4,5,8,10$), where
$\C[\su3] = 2\,\hat Q\cdot\hat Q + \tfrac34 \hat L^2$ with the $\su3$ quadrupole
$\hat Q_\mu = s^{\dagger}\tilde d_\mu + d^{\dagger}_\mu s - \tfrac{\sqrt7}{2}(d^{\dagger}\tilde d)^{(2)}_\mu$.
Equation~\eqref{eq:arbiter} is used as an unambiguous arbiter of Clebsch--Gordan and
tensor-normalization conventions.

For the critical-point analysis we use Leviatan's constructions~\cite{Leviatan2007}. The
first-order (spherical$\leftrightarrow$prolate) critical Hamiltonian is
\begin{align}
\hat H(\beta_0) &= h_2\, \Pdag_2(\beta_0)\cdot\tilde P_2(\beta_0), \label{eq:Hfo}\\
\Pdag_{2\mu}(\beta_0) &= \beta_0\, s^{\dagger} d^{\dagger}_\mu
+ \sqrt{\tfrac72}\,(d^{\dagger}d^{\dagger})^{(2)}_\mu , \nonumber
\end{align}
whose kernel is a deformed ground band for all $\beta_0$; at the special value $\beta_0=\sqrt2$ it
acquires, in addition, the full solvable $\su3$ tower, and $\hat H(\sqrt2)=\tfrac12\times[\,h_2$-term
of Eq.~\eqref{eq:H}$]$ (verified to $3\times10^{-14}$), tying it to the stable Hamiltonian. The
second-order (spherical$\leftrightarrow\gamma$-unstable) critical Hamiltonian is
\begin{equation}
\hat H = \varepsilon\,\nd + A\big[(d^{\dagger}\!\cdot d^{\dagger} - (s^{\dagger})^2) + \mathrm{h.c.}\big],
\quad \varepsilon = 4(N-1)A ,
\label{eq:Hso}
\end{equation}
which is an $\oo5$ scalar: $\|[\hat H, \C[\oo5]]\|_F = 2.7\times10^{-12}$ (Frobenius norm, full fixed-$N$ space) in our implementation, so
that every eigenstate retains an exact $\oo5$ seniority $\tau$.

\section{The measurable: label variance, block coherence and block purity}
\label{sec:measurable}

Let $G$ be a subalgebra of the spectrum-generating algebra with quadratic Casimir $\C[G]$. Its
eigenvalues $f_2(\lambda)$ are constant on each irrep $\lambda$ contained in the model space and
give an orthogonal decomposition of the Hilbert space into $\C[G]$ eigenspaces,
$\mathcal H = \bigoplus_a \mathcal H_a$, with projectors $\Pi_a$ onto the value $a$. For a
normalized pure state $\ket\psi$ define the sector-probability distribution
$P_a(\psi) = \bra\psi \Pi_a \ket\psi = \|\Pi_a\ket\psi\|^2$, and the three quantities
\begin{align}
\text{label variance:}\quad & \Var_\psi(\C) = \bra\psi \C^2 \ket\psi - \bra\psi \C \ket\psi^2, \\
\text{block coherence:}\quad & S_G(\psi) = -\sum_a P_a \ln P_a, \\
\text{block purity:}\quad & \mathcal P_G(\psi) = \sum_a P_a^2 = \Tr(\rho_G^2),
\end{align}
where $\rho_G = \Delta_G(|\psi\rangle\langle\psi|) = \sum_a \Pi_a |\psi\rangle\langle\psi| \Pi_a$ is the sector-dephased state produced by the completely dephasing map $\Delta_G$ with respect to the Casimir eigenprojectors $\Pi_a$, so that $P_a=\langle\psi|\Pi_a|\psi\rangle$ and $\rho_G=\sum_a P_a\ket{a}\!\bra{a}$ when $|\psi\rangle$ has a single component per sector; $S_G$ is the relative entropy of coherence associated with $\Delta_G$, and for a pure state reduces to the entropy of $\{P_a\}$, with any coherence internal to a sector left untouched by $\Delta_G$ (natural logarithm throughout).
Throughout, ``block purity'' denotes this
symmetry-sector purity $\Tr(\rho_G^2)$, a property of the sector decomposition, and should not be
confused with the ordinary state purity $\Tr(\rho^2)=1$ of a pure state.

The three quantities share a single zero condition, recorded in the elementary statement below.
Its point is not the proof but that one physical condition, support in a single Casimir sector, is
read simultaneously as a variance, a coherence, and a purity; this is what lets us make the
``purity/coherence'' language of Sec.~\ref{sec:kremer} literal.

\begin{proposition}
\label{prop:equiv}
For any pure state $\ket\psi$ the following are equivalent:
(a) $\ket\psi$ lies in a single $\C[G]$ eigenspace $\mathcal H_a$;
(b) $\Var_\psi(\C[G]) = 0$;
(c) $S_G(\psi) = 0$;
(d) $\mathcal P_G(\psi) = 1$.
If, in addition, the map $\lambda\mapsto f_2(\lambda)$ is injective on the irreps present in the
model space, then each of (a)--(d) is equivalent to $\ket\psi$ carrying a single exact $G$-irrep
label.
\end{proposition}

\begin{myproof}
Since $\C[G]$ acts as the scalar $a$ on $\mathcal H_a$,
$\avg{\C} = \sum_a P_a\,a$ and $\avg{\C^2} = \sum_a P_a\,a^2$, so
\begin{equation}
\Var_\psi(\C) = \sum_a P_a\,(a - \avg{\C})^2 ,
\label{eq:varid}
\end{equation}
the variance of the value $a$ under $\{P_a\}$. Each of (b) the vanishing of Eq.~\eqref{eq:varid},
(c) the vanishing of the Shannon entropy $S_G$, and (d) the saturation
$\mathcal P_G=\sum_a P_a^2=1$ holds iff $\{P_a\}$ is a point mass, i.e.\ iff $\ket\psi$ occupies a
single eigenspace~(a). A single eigenspace is a single irrep exactly when no two inequivalent
irreps of the model space share the value $f_2(\lambda)$, which is the stated injectivity. We
verified Eq.~\eqref{eq:varid} numerically to $3\times10^{-10}$.
\end{myproof}

The injectivity proviso is not automatic. For $\su3$ the quadratic Casimir
$f_2(\lambda,\mu)=\lambda^2+\mu^2+\lambda\mu+3(\lambda+\mu)$ is symmetric under
$(\lambda,\mu)\leftrightarrow(\mu,\lambda)$, so conjugate irreps are $\C[\su3]$-degenerate. In the
symmetric $\uu6$ irrep $[N]$ such conjugate pairs occur whenever $N$ is a multiple of three; within
the range used here they appear at $N=6,9,12,15$ (for example $(6,0)$ and $(0,6)$ at $N=6$, four
pairs at $N=12$). At those $N$ a vanishing $\Var\,\C[\su3]$ certifies support in one $\C[\su3]$
eigenspace, which may be a two-irrep space rather than a single irrep. The independent invariant
that separates conjugate irreps is the cubic Casimir $\Cthree[\su3]$, whose eigenvalue
$g_3(\lambda,\mu)=(\lambda-\mu)(2\lambda+\mu+3)(\lambda+2\mu+3)$ is antisymmetric under conjugation
and therefore nonzero on every conjugate pair. We constructed a Hermitian $\Cthree[\su3]$ from the
coupled products of $\hat Q$ and $\hat L$ (the explicit operator and its coefficient are given in the
accompanying code; it is Hermitian by construction and commutes with all $\su3$ generators to
$\lesssim10^{-12}$ in Frobenius norm) and rechecked every $\su3$ classification below with the joint pair
$(\C,\Cthree)$: at all $N$ studied, and at both the stable and first-order points, every state with
$\Var\,\C[\su3]=0$ also has $\Var\,\Cthree[\su3]=0$, so no state is a conjugate-irrep mixture
masquerading as symmetry-pure, and no count changes. The $\oo6$, $\oo5$ and $\uu5$ labels used
elsewhere are read off $\sigma(\sigma+4)$, $\tau(\tau+3)$ and $\nd$, all injective, so the proviso
is vacuous there.

Proposition~\ref{prop:equiv} identifies exact irrep purity with a QI condition. It is a statement
about sector membership, not about solvability: a Hamiltonian carrying an isolated exact-label
eigenstate need not be a PDS Hamiltonian, and, conversely, not every exact-label eigenstate of a
PDS Hamiltonian is one of its closed-form solvable states (Sec.~\ref{sec:crit} exhibits both
sides). Within a Hamiltonian independently known to possess a PDS structure, the vanishing sector
fluctuation identifies the symmetry-preserving eigenstates directly from their wave functions; the
solvability and algebraic construction are established separately, by Leviatan's annihilation
conditions and closed-form spectrum. Accordingly we take $\Var\,\C[G]=0$ throughout as the exact
statement of \emph{irrep purity}, and use the closed-form energy $E = 6h_2 k(2N-2k+1)+C\,L(L+1)$ as
the separate criterion for \emph{solvability} when the two need to be distinguished.

Although the three faces share a zero, they carry different information away from it and differ in
what it takes to evaluate them, and this is why we adopt the variance as the primary diagnostic.
The variance $\Var_\psi(\C[G])$ requires only two expectation values, $\avg{\C[G]}$ and
$\avg{\C[G]^2}$, so it is basis-free and needs no explicit sector decomposition: the Casimir is
measured directly, on prepared states as well as computed ones, but it weights each sector by the
squared separation of its Casimir value from the mean. The entropy $S_G$ and the purity
$\mathcal P_G$ depend only on the probabilities $\{P_a\}$ and require the full sector distribution,
i.e.\ the projections onto every irrep, which is more information and more work to obtain. The three
therefore agree on whether a state is exactly confined to one sector but do not carry identical
information for a mixed state; we report the variance throughout and use $S_G$ when a
coherence-monotone framing is wanted, or $\mathcal P_G$ when a purity framing is wanted, the latter
making the ``purity'' and ``coherence'' language of QDS literal on the same footing
(Sec.~\ref{sec:kremer}). We stress that $\Var(\C[G])$ and the bipartite $s|d$ entropy are functions
of \emph{different} structures, the $G$-sector decomposition versus the $s$--$d$ tensor
factorization, so they are logically distinct, being functions of different structures rather than
functions of each other.

At symmetric endpoints $H$-degeneracies leave the label undefined within a degenerate eigenspace.
We resolve this with an explicit, order-fixed refinement applied within each cluster of levels
lying inside an energy tolerance $\delta_E=10^{-9}$ (in units of the spectrum scale): we diagonalize
$\hat L^2$ first, then $\C[G]$, then $\nd$, restricted to the cluster, before any label variance is
computed. All three operators are mutually commuting on the physical space, so on an exact
degenerate eigenspace the refinement returns basis-independent sector projectors; the residual
freedom is only the choice among tie-breakers for states that are genuinely degenerate in all of
$L^2$, $\C[G]$ and $\nd$, which we note where it affects the labelling of individual states below.
The total count of exact-label states in a degenerate cluster is a property of the subspace and is
independent of this order. All ``pure versus mixed'' classifications use the threshold
$\Var\,\C[G]<10^{-6}$; the classification is insensitive to this choice, since at a stable $\su3$
point ($N=10$) the largest exact-label variance is $6\times10^{-10}$ while the smallest nonzero one
is $24$, and the count is unchanged for any threshold from $10^{-8}$ to $1$ (Table~\ref{tab:thresh}). We write ``analytically zero'' for exact algebraic results, ``numerically consistent with zero'' (or ``at the diagonalization floor,'' $\lesssim10^{-9}$) for computed values, and ``$<10^{-6}$'' for the classification threshold.
Bipartite entropies use the $s|d$ \emph{mode} bipartition: the six single-boson modes are split into
the $s$ mode and the five $d$ modes, the reduced density matrix is the partial trace of
$|\psi\rangle\langle\psi|$ over the $d$ modes, and $S_{s|d}=-\Tr\rho_s\ln\rho_s$ (unnormalized, natural
log). All states compared carry the same $N$ and $M=0$. This is a mode entanglement, not a particle
entanglement, and it carries the usual particle-number-superselection caveat for boson-mode
entanglement, which restricts the operational meaning of the magnitude but does not affect the
label-variance results, which are a property of the $G$-sector decomposition rather than of any mode
factorization. Accordingly we claim only that, for this bipartition and in the examples studied, the
$s|d$ von Neumann entropy does not separate the symmetry-pure and mixed eigenstates; we do not claim it
is impossible for some other partition.

\begin{table}[t]
\caption{Threshold robustness of the pure/mixed classification at a stable $\su3$-PDS point
($N=10$, $203$ states in the $M=0$ block). The number of states with $\Var\,\C[\su3]$ below the
cut is constant over eight decades; it drops only at $10^{-10}$, the numerical noise floor of the
degeneracy refinement.}
\label{tab:thresh}
\begin{ruledtabular}
\begin{tabular}{lcccccc}
threshold & $10^{-10}$ & $10^{-8}$ & $10^{-6}$ & $10^{-4}$ & $10^{-2}$ & $1$ \\
\# exact-label & 50 & 58 & 58 & 58 & 58 & 58 \\
\end{tabular}
\end{ruledtabular}
\end{table}

\section{Stable-point results}
\label{sec:stable}

\begin{figure}[t]
\centering
\includegraphics[width=\columnwidth]{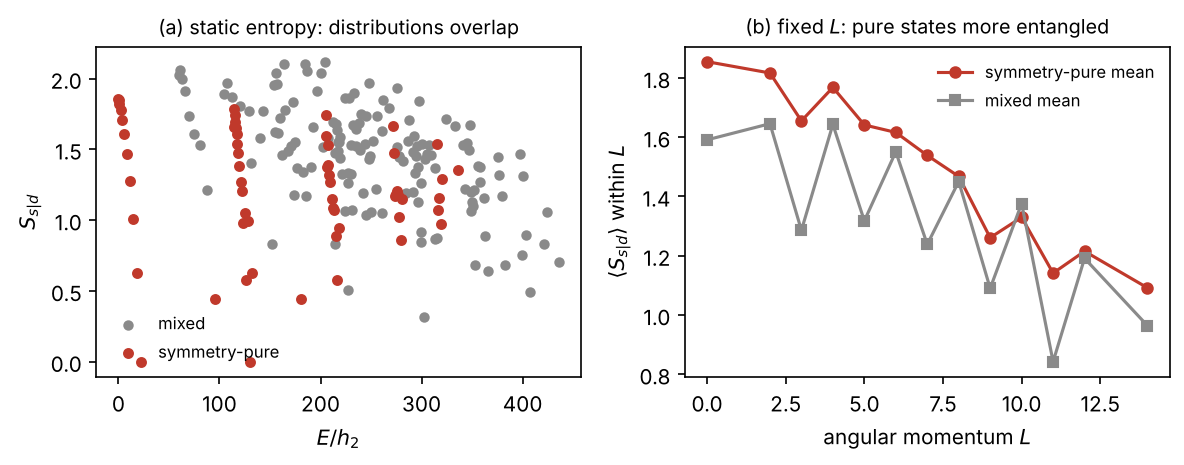}
\caption{Decision gate at a stable $\su3$-PDS point ($N=10$). The static \emph{magnitude} of the
$s|d$ entanglement entropy does not separate symmetry-pure (exact-label) from mixed eigenstates:
the distributions overlap, and within a fixed angular momentum the pure states are in fact
\emph{more} entangled (a Simpson's-paradox reversal). Entanglement magnitude alone is therefore
not a PDS diagnostic.}
\label{fig:gate}
\end{figure}

\begin{figure}[t]
\centering
\includegraphics[width=\columnwidth]{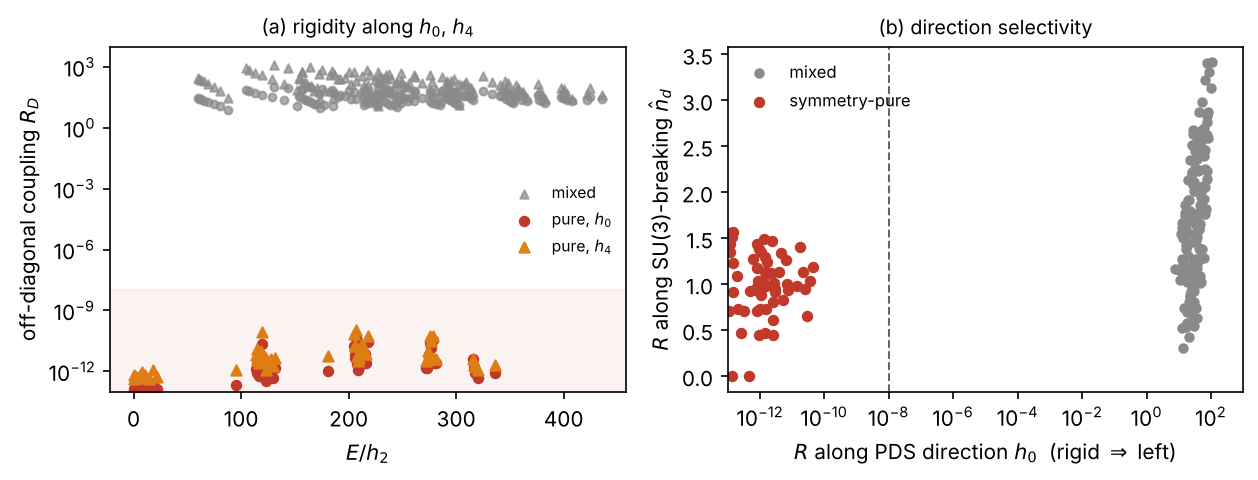}
\caption{Cone rigidity at a stable $\su3$-PDS point ($N=10$). The off-diagonal coupling
$R_D=\|Q\,D\ket\psi\|$ measures the first-order response of each eigenstate to a deformation along
direction $D$. Along the PDS-preserving directions $h_0$ and $h_4$ every symmetry-pure state is
rigid to numerical zero ($R\lesssim3\times10^{-10}$, lower band) while every mixed state responds
strongly ($R\gtrsim70$, upper band); the separation is direction-selective, since along the
$\su3$-breaking direction $\nd$ the pure states respond as much as the mixed ones ($R\sim1$,
right). Rigidity is a property of the PDS-preserving operator cone, not of the states in isolation.}
\label{fig:cone}
\end{figure}

We first test the naive expectation that symmetry-pure states are simply less entangled.
Figure~\ref{fig:gate} shows they are not: at a stable $\su3$-PDS point the static $s|d$ entanglement
entropy does not separate the pure subset from the mixed states, and within a fixed $L$ the pure
states are, if anything, \emph{more} entangled. This negative result matters, because it rules out
reading any later separation as ``symmetry-pure states are simply simpler.''

The signature instead lives in the response of the eigenstates to a symmetry-preserving
deformation. Consider the two-parameter family of PDS-preserving operators (the coefficients
$h_0$ and a higher-order term $h_4$ that both leave the solvable subset intact). We measure the
first-order response of each eigenstate, the off-diagonal coupling $R_D=\|Q\,D\ket\psi\|$ with $Q$
projecting out the state's own degenerate eigenspace, which vanishes exactly when the state is rigid
along direction $D$. Along the PDS-preserving directions every solvable state is rigid to numerical
zero ($R\lesssim3\times10^{-10}$ at $N=8,10$), while every mixed state responds strongly
($R\gtrsim70$), a separation of more than ten orders of magnitude with zero misclassifications
(Fig.~\ref{fig:cone}, $N=8,10$). The rigidity is \emph{direction-selective}: it holds
along $h_0$ \emph{and} $h_4$, the two independent PDS-preserving directions, but is lost along the
$\su3$-breaking direction $\nd$, where the solvable states respond as much as the mixed ones
($R_{\nd}\sim1$ for both). This
direction selectivity answers the objection that rigidity is a tautology (``solvable states don't
change because they are eigenstates''): the same states are mobile along a symmetry-breaking
direction, so rigidity is a property of the PDS-preserving cone, not of the states in isolation.
The solvable subspace is invariant under the PDS-preserving flow (subspace overlap $1.000000$ at
$N=12$), confirming the result is not a small-$N$ accident.

\section{Criticality}
\label{sec:crit}

\subsection{PDS is not generic to the quantum phase transition}

\begin{figure*}[t]
\centering
\includegraphics[width=0.9\textwidth]{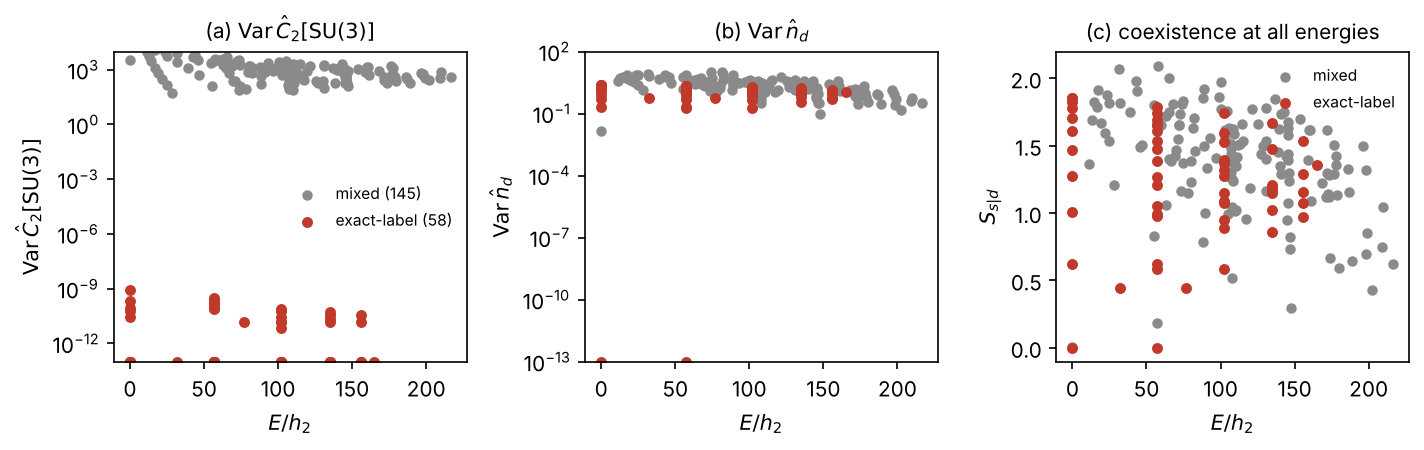}
\caption{First-order critical point, $\hat H(\beta_0=\sqrt2)$, $N=10$ (type~I coexistence). Left:
the label variance $\Var\,\C[\su3]$ resolves the spectrum into an exact-label subset
($\Var=0$, $58$ states) and a doubly-mixed background ($\Var\gg0$, $145$ states); the exact-label
subset splits into $56$ closed-form solvable band members and $2$ irrep-pure but non-solvable states
(see text). The deformed $L=0$ band head sits at $\avg{\nd}/N=0.63$. Centre: the same separation
seen in $\Var\,\nd$. Right: the exact-label subset coexists with the mixed states at all excitation
energies, not only near the ground state.}
\label{fig:q2}
\end{figure*}

Before turning to Leviatan's critical Hamiltonians we establish that PDS is a
special construction, not an automatic property of any IBM QPT. On the canonical consistent-$Q$
line $\hat H(\xi) = (1-\xi)\nd - (\xi/4N)\,\hat Q\cdot\hat Q$ we reproduce the textbook first-order
QPT, the fidelity susceptibility peaks at $\xi_c=0.500$ and $\avg{\nd}/N$ jumps from $0$ to
$0.66$, yet we find \emph{zero} nontrivial exact-label states in the interior for every
$\xi\in\{0.10,\dots,0.90\}$ (excluding the trivial $n_d=0$ and $n_d=N$ multiplets, which are exact
at all $\xi$). Exact symmetry appears only at the DS endpoints and via Leviatan's off-line PDS
construction. This is the sharp statement that separates ``PDS at criticality'' from ordinary
critical-point phenomenology, and it pre-empts the objection that PDS is merely a relabelling of a
transition.

\subsection{First order: a solvable subset amid strong mixing (type~I)}

At the first-order critical Hamiltonian $\hat H(\beta_0=\sqrt2)$ of Eq.~\eqref{eq:Hfo}, the
spectrum splits sharply (Fig.~\ref{fig:q2}). For $N=10$ (a $203$-dimensional $M=0$ block), $58$
exact-label states coexist with $145$ doubly-mixed states, realizing PDS type~I coexistence
numerically: the exact-label subset has $\Var\,\C[\su3]$ or $\Var\,\nd$ zero to numerical precision
while the background variance is large; its scaling is quantified in Sec.~\ref{sec:scaling}.

The $58$ exact-label states are not all closed-form solvable, and the distinction is instructive.
Fifty-six of them are the analytic $\su3$ solvable tower: the ground band $(2N,0)$ and the
$\gamma^k$ bands $(2N-4k,2k)$, whose members obey the closed-form energy
$E=6h_2 k(2N-2k+1)+C\,L(L+1)$ to machine precision. The remaining two carry an exact $\su3$ label
but do \emph{not} obey that formula. They arise because a solvable irrep can appear in the $M=0$
block with $K$-band multiplicity greater than one: at $N=10$ the irrep $(16,2)$ contains $L=16$ in
both its $K=2$ band (Leviatan's solvable $\gamma^1$ member) and its $K=0$ band, and the irrep
$(12,4)$ contains $L=13$ with multiplicity two. Leviatan's construction makes only one member of
each such pair solvable, but both members lie entirely within the single irrep, so the non-solvable
partner still has $\Var\,\C[\su3]=0$ while its energy departs from the closed form. This is exactly
what the diagnostic should report: $\Var\,\C[\su3]=0$ certifies irrep purity, a strictly weaker
property than closed-form solvability, and the closed-form energy check separates the two tiers
cleanly. Both non-solvable partners are, like the solvable band members, rigid along the
PDS-preserving cone, so the classification of Sec.~\ref{sec:stable} is unaffected. The count of
exact-label states is a subspace property and does not depend on the degeneracy-resolution order;
what the order can move is only the assignment of the two U(5) solvable states of Leviatan's
construction ($|n_d=\tau=L=0\rangle$ and $|n_d=\tau=L=3\rangle$) between the $\su3$ and $\uu5$
readings within their shared level, which is why we report the exact-label total rather than a
split.

The closed-form solvable tower follows the $(2N-4k,2k)$ band-projection count,
$22, 37, 56, 79$ for $N=6,8,10,12$ (exact), so that its fraction of the $M=0$ space shrinks from
about $0.5$ to $0.2$: the diagnostic \emph{sharpens} as $N$ grows, the opposite of a finite-size
artifact. The exact-label total (tower plus the multiplicity partners) is $24, 39, 58, 81$ over the
same $N$. Table~\ref{tab:classes} makes the three tiers explicit: Class~I, exact-label and closed-form
solvable; Class~II, exact-label but not closed-form solvable (the multiplicity partners); and
Class~III, mixed. Class~II holds steady at two states for every $N\ge6$, so it is a structural feature
of the construction rather than a small-$N$ accident.

\begin{table}[t]
\caption{State counts by tier in the $M=0$ block at the SU(3)-PDS point. Class~I: exact-label and
closed-form solvable (the $(2N-4k,2k)$ tower). Class~II: exact-label but not closed-form solvable (the
multiplicity partners). Class~III: mixed. Classes I$+$II reproduce the exact-label totals of
Table~\ref{tab:scaling}.}
\label{tab:classes}
\begin{ruledtabular}
\begin{tabular}{rrrrr}
$N$ & $d_{M=0}$ & Class I & Class II & Class III \\
\colrule
4  & 18  & 11  & 1 & 6 \\
6  & 48  & 22  & 2 & 24 \\
8  & 105 & 37  & 2 & 66 \\
10 & 203 & 56  & 2 & 145 \\
12 & 358 & 79  & 2 & 277 \\
14 & 588 & 106 & 2 & 480 \\
16 & 915 & 137 & 2 & 776 \\
\end{tabular}
\end{ruledtabular}
\end{table}

\subsection{Entanglement magnitude is not the criticality signature}

The deformed ground band is the kernel of $\hat H(\beta_0)$ for all $\beta_0$ and is therefore
both eigenvalue- and label-rigid. Its wavefunction, however, is projected from the
$\beta_0$-dependent intrinsic condensate and \emph{rotates} with $\beta_0$: the kernel at one
$\beta_0$ has nonzero principal angles with the kernel at a neighbouring value, so the ground band
is genuinely $\beta_0$-dependent. Consequently its $s|d$ entropy \emph{drifts} smoothly with
$\beta_0$ even though its labels never change. The stable-point cone rigidity
does not transfer to the critical control parameter. The surviving, sharp signature is
$\Var\,\C[\su3]=0$ on the exact-label subset versus large on the mixed states, a label statement,
not an entanglement-magnitude statement.

\subsection{Second order: one shared label and a seniority ladder (type~II)}

\begin{figure*}[t]
\centering
\includegraphics[width=0.9\textwidth]{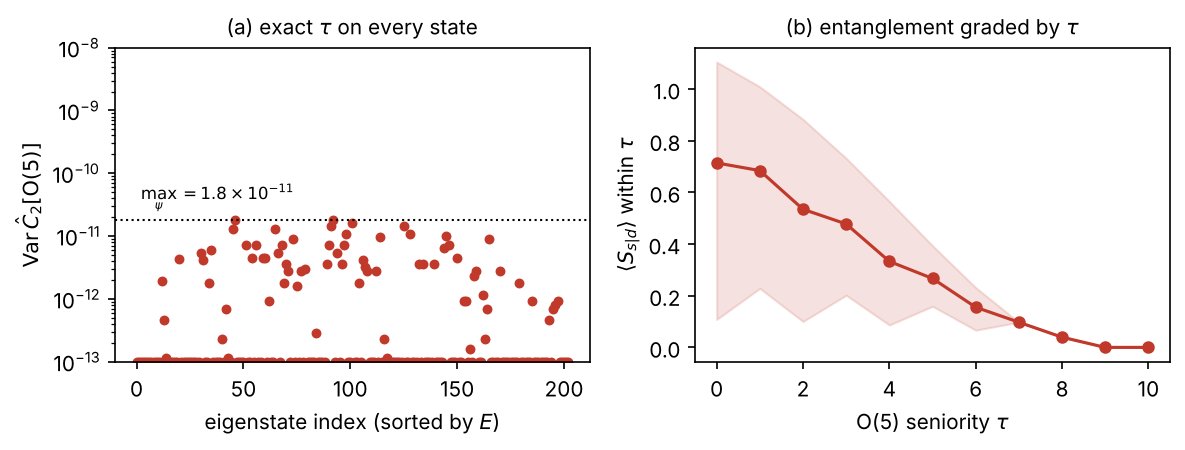}
\caption{Second-order critical point, Eq.~\eqref{eq:Hso}, $N=10$ (type~II). Every eigenstate
carries an exact $\oo5$ seniority: $\max_\psi \Var\,\C[\oo5] = 2.7\times10^{-11}$ over the whole
spectrum (left). The $s|d$ entanglement is graded monotonically by $\tau$ (right), $\avg{S}=0.71$
at $\tau=0$ falling to $0.00$ at $\tau=N$: the entanglement lives in the low-seniority sector and
high-seniority states are near-product. This seniority-graded ladder differs qualitatively from the
first-order solvable-subset picture of Fig.~\ref{fig:q2}.}
\label{fig:q2so}
\end{figure*}

The second-order critical Hamiltonian of Eq.~\eqref{eq:Hso} is an $\oo5$ scalar, so \emph{every}
eigenstate keeps an exact seniority $\tau$: the label variance $\Var\,\C[\oo5]$ is at most
$2.7\times10^{-11}$ over the entire spectrum (Fig.~\ref{fig:q2so}). Here the $\oo5$ label is read
from $\tau(\tau+3)$, which is injective, so the exact-label statement is unambiguous. This is
the type-II PDS pattern associated with a conserved $\oo5$ label: not a solvable subset but a single conserved label shared by all states (type~II in Leviatan's classification~\cite{Leviatan2011}, where one subalgebra labels every eigenstate while a competing one does not). The label
organizes the entanglement into a seniority-graded ladder, the mean $s|d$ entropy falling
monotonically from $0.71$ at $\tau=0$ to $0.16$ at $\tau=6$ to $0.00$ at $\tau=N$, so that
high-seniority states are near-product and the entanglement is concentrated in the low-$\tau$
sector. The QI fingerprint here differs in kind from the first-order solvable-subset picture, and
that difference tracks the order of the transition.

\section{Robustness, scaling, and specificity}
\label{sec:robust}

Three questions decide whether the diagnostic is more than a small-$N$ curiosity: how the separation
scales with boson number, how it degrades when the symmetry is only approximate, and whether it fires
on Hamiltonians that are not of partial-dynamical-symmetry type.

\subsection{System-size scaling}
\label{sec:scaling}

Table~\ref{tab:scaling} collects the separation at the stable $\su3$-PDS point from $N=4$ to $N=16$.
The closed-form solvable count follows the band-projection prediction ($137$ at $N=16$), the total
exact-label count being larger by the two multiplicity partners ($139$); the exact-label
\emph{fraction} shrinks from $0.67$ to $0.15$, so the subset the diagnostic selects becomes a
progressively smaller part of the spectrum. The pure-state residual stays at the level of the
diagonalization error ($\lesssim 3\times10^{-9}$), while the smallest mixed-state variance stays
bounded away from zero ($\gtrsim 22$) at every $N$. The separation ratio, the smallest mixed variance
divided by the largest pure residual, therefore exceeds $10^{9}$ throughout and does not close; since its
denominator is set by the diagonalization floor it is a numerical classification ratio rather than a
physical ratio, and we report it alongside the absolute values above. The
mixed-state variance is not a clean $N^2$: the median grows as $N^{2.3}$ by a log-log fit and the
distribution widens toward $N^4$ at the top, consistent with a quadratic Casimir whose sector spacing
is itself $O(N^2)$. The robust statement is thus the one used throughout the text, that the pure states
are numerically zero and the mixed states are large and separated from zero by a margin that does not
close, rather than a specific power law.

\begin{table}[t]
\caption{Scaling of the label variance at the stable $\su3$-PDS point. ``frac'' is the exact-label
fraction, ``max residual'' the largest $\Var\,\C[\su3]$ among pure states (diagonalization floor),
``min mixed'' and ``median mixed'' the smallest and median values over the mixed states. The
separation ratio (min mixed / max residual) exceeds $10^{9}$ at every $N$.}
\label{tab:scaling}
\begin{ruledtabular}
\begin{tabular}{rrrrrrr}
$N$ & $d_{M=0}$ & \#pure & frac & max residual & min mixed & median mixed \\
\colrule
4  & 18  & 12  & 0.67 & $3\times10^{-12}$ & 18.1 & 61.7 \\
6  & 48  & 24  & 0.50 & $2\times10^{-11}$ & 28.0 & 125.3 \\
8  & 105 & 39  & 0.37 & $5\times10^{-11}$ & 25.7 & 242.3 \\
10 & 203 & 58  & 0.29 & $4\times10^{-10}$ & 24.3 & 369.0 \\
12 & 358 & 81  & 0.23 & $5\times10^{-10}$ & 23.3 & 601.7 \\
14 & 588 & 108 & 0.18 & $1\times10^{-9}$  & 22.6 & 1035.0 \\
16 & 915 & 139 & 0.15 & $3\times10^{-9}$  & 22.1 & 1577.1 \\
\end{tabular}
\end{ruledtabular}
\end{table}

\begin{figure}[t]
\centering
\includegraphics[width=\columnwidth]{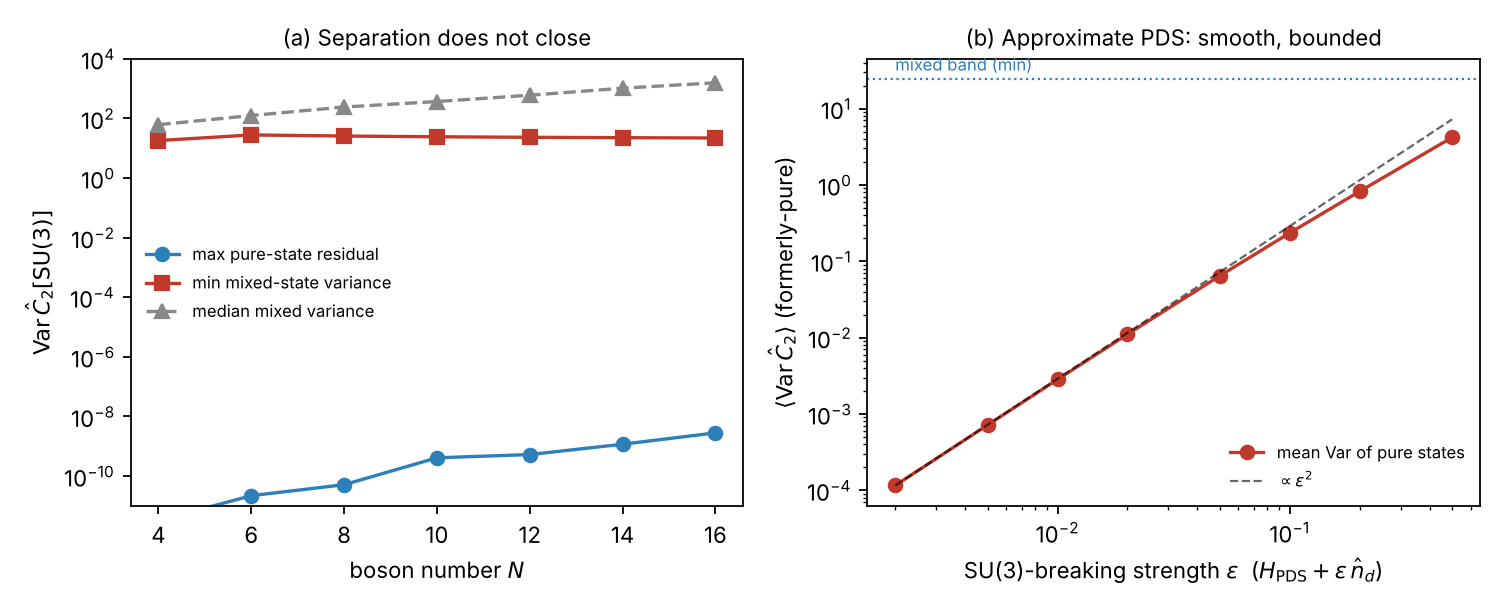}
\caption{(a) The separation does not close, shown as absolute values: the maximum pure-state residual
(at the diagonalization floor, $\lesssim3\times10^{-9}$), the minimum mixed-state variance ($\gtrsim22$),
and the median mixed-state variance, from $N=4$ to $16$. The pure and mixed bands stay apart at every
$N$; their ratio is a secondary numerical classification ratio (denominator set by the diagonalization
floor), not a physical one. (b) Approximate
PDS: under $\hat H_{\rm PDS}+\epsilon\,\hat n_d$ the formerly-pure states acquire a label variance that
grows as $\epsilon^2$ (dashed guide) and stays below the mixed band (dotted) up to a breaking of order
ten percent.}
\label{fig:robust}
\end{figure}

\subsection{Approximate partial dynamical symmetry}
\label{sec:approxpds}

Real nuclei realize PDS only approximately, so the practical question is how the diagnostic behaves
under a controlled breaking. We add an $\su3$-breaking term, $\hat H(\epsilon)=\hat H_{\rm PDS}
+\epsilon\,\hat n_d$, and follow each $\epsilon=0$ symmetry-pure state to the instantaneous eigenstate
of largest overlap as $\epsilon$ grows, which tracks it through the avoided crossings that open along
the way (Fig.~\ref{fig:robust}(b), $N=10$). The mean label variance over this set rises from numerical
zero as $\langle\Var\,\C[\su3]\rangle\sim\epsilon^{2}$; a log-log fit over $\epsilon\le0.05$, in the
dimensionless units in which the stable Hamiltonian has unit $h_2$, gives an exponent of $1.97$, as
first-order perturbation theory predicts, since first-order admixtures produce second-order leakage of
probability into other sectors. A second breaking direction confirms the growth is not special to
$\hat n_d$: we instead add the $\oo5$ Casimir $\C[\oo5]$, which does not commute with $\C[\su3]$ and so
lifts the $\su3$ labelling while preserving the rotational $\oo3$ symmetry, and which is linearly
independent of $\hat n_d$ as an operator. Applied in the same dimensionless units, it gives a
small-$\epsilon$ exponent of $1.8$ over the $\epsilon\le0.02$ window. Both exponents are finite-window
fits consistent with the perturbative $\epsilon^2$ law, not claimed asymptotic values. The pure-state variance stays below the smallest mixed-state value for perturbations up to
of order ten percent in this normalization, so the pure/mixed separation remains visible across that
window and the variance becomes a continuous measure of how close a fitted Hamiltonian sits to the
symmetry limit.

The same controlled breaking also lets us compare the three faces of Sec.~\ref{sec:measurable}
quantitatively, rather than only noting that they share a zero. At $\epsilon=0.02$ ($N=10$) the
formerly-pure states acquire, averaged over the set, a label variance $\Var\,\C[\su3]=3.0\times10^{-3}$,
a block-coherence entropy $S_{\su3}=3.7\times10^{-6}$ and a block impurity $1-\mathcal P_{\su3}=
4.5\times10^{-7}$; all three rise smoothly from numerical zero together, so none is an artifact of a
particular measure. Their absolute sizes differ because they weight the leaked probability differently:
the variance multiplies each sector's small probability by the squared separation of its Casimir
eigenvalue from the mean, a separation of order $N^4$ here, whereas the entropy and purity weight only
by the probability itself. The practical distinction is in what must be measured. The variance is fixed
by the two moments $\avg{\C[\su3]}$ and $\avg{\C[\su3]^2}$ alone, which is why it transfers directly to a
prepared state, while the entropy and the purity require the full sector distribution $\{P_a\}$, i.e.\ a
projection onto every $\su3$ irrep. We therefore report the variance as the primary diagnostic and use
the entropy or purity only when a coherence- or purity-monotone framing is wanted.

\subsection{Specificity on non-PDS Hamiltonians}
\label{sec:specificity}

The diagnostic identifies the symmetry-selected states of a Hamiltonian already of PDS type; it is not
a universal detector that fires on an arbitrary Hamiltonian. Table~\ref{tab:counts} collects the
$\su3$-exact-label count ($\Var\,\C[\su3]<10^{-6}$ in the $N=10$, $M=0$ block) across representative
Hamiltonians. At the $\su3$ dynamical-symmetry vertex every one of the $203$ states carries an exact
label, as it must in the exact-symmetry limit (this count refers to the degeneracy-refined basis of
Sec.~\ref{sec:measurable}; an arbitrary numerical basis within a degenerate energy eigenspace need not
itself carry a definite $\su3$ label); at the $\uu5$ and $\oo6$ vertices and at the
second-order critical point none does, the conserved label there being $\oo5$ seniority rather than
$\su3$; the stable PDS point has $58$ and the first-order critical point $56$ (together with the two
$\uu5$-pure states of Sec.~\ref{sec:crit}). At a generic interior consistent-$Q$ point
$\hat H(\xi,\chi)=(1-\xi)\hat n_d-(\xi/4N)\,\hat Q_\chi\!\cdot\!\hat Q_\chi$ only the $n_d=0$ and
$n_d=N$ multiplets, which carry an exact label everywhere, remain. Scanning the interior at twelve
random points with $\xi$ drawn uniformly from $[0.05,0.95]$ and $\chi$ from $[-\sqrt7/2,0]$, i.e.\
bounded away from the three dynamical-symmetry vertices (seed and script in the repository), no
nontrivial $\su3$-exact-label state was found at any sampled point; the smallest nontrivial label
variance encountered was $\approx2\times10^{2}$, about eight orders of magnitude above the $10^{-6}$
classification threshold. We state this as a finite numerical negative control: no false positives
were observed in the tested interior sample of the consistent-$Q$ family, extending the single
consistent-$Q$ line of Sec.~\ref{sec:crit}. The sample spans the $\uu5$--$\oo6$--$\su3$ interior but,
by construction, not the immediate neighborhoods of the vertices, where the exact labels reappear.

\begin{table}[t]
\caption{Number of $\su3$-exact-label states ($\Var\,\C[\su3]<10^{-6}$) in the $N=10$, $M=0$ block
(dimension $d_{M=0}=203$). ``Trivial only'' denotes the $n_d=0$ and $n_d=N$ multiplets that carry an
exact label at any parameter point. At the first-order critical point the two further exact-label
states are $\uu5$-pure states sharing a degenerate level (Sec.~\ref{sec:crit}); they are not $\su3$
labels and so are not included in the $\su3$-exact-label count of $56$.}
\label{tab:counts}
\begin{ruledtabular}
\begin{tabular}{lr}
Hamiltonian & \#\,$\su3$-exact-label \\
\colrule
SU(3) dynamical-symmetry vertex & 203 (all) \\
stable SU(3)-PDS point & 58 \\
first-order critical point & 56 \\
generic consistent-$Q$ interior & trivial only \\
U(5) vertex & 0 \\
O(6) vertex & 0 \\
second-order critical point & 0 \\
\end{tabular}
\end{ruledtabular}
\end{table}

\section{Bridge to quasi-dynamical symmetry: purity and coherence}
\label{sec:kremer}

\begin{figure*}[t]
\centering
\includegraphics[width=0.9\textwidth]{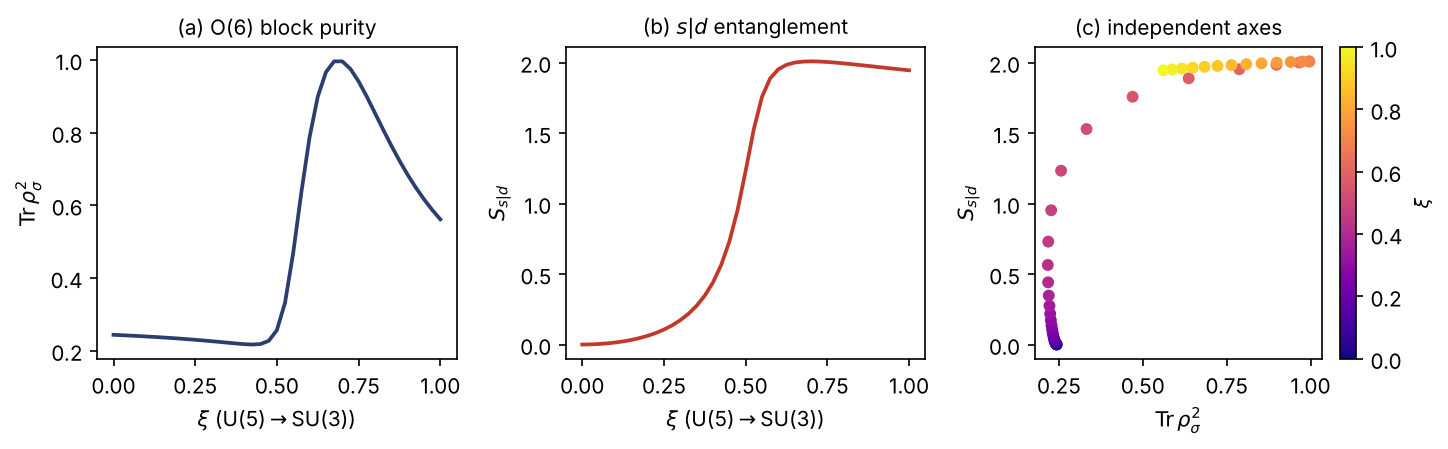}
\caption{The purity/coherence bridge on Kremer \emph{et al.}'s own states, $N=12$. The $\oo6$
block purity $\Tr\rho_\sigma^2$ (left) and the bipartite $s|d$ entanglement entropy (centre) are
structurally distinct, non-monotonically related axes: the ground state becomes $\oo6$-\emph{pure}
($\Tr\rho_\sigma^2 = 0.81$--$0.88$) precisely where its $s|d$ entanglement is \emph{maximal} (right).
The algebraic ``purity'' of the quasi-dynamical-symmetry literature and bipartite entanglement are
distinct QI quantities, functions of different structures (the sector decomposition versus the boson
bipartition). Our $\oo6$ machinery reproduces Kremer's $\sigma$-fluctuation $\Delta\sigma_{\rm gs}=2.473$
at the $\uu5$ limit (their value $2.47$) and the exactly sharp $\oo6$ point.}
\label{fig:kremer}
\end{figure*}

Kremer \emph{et al.}~\cite{Kremer2014} identified a region of IBM parameter space possessing both
$\oo6$-PDS, which they labelled ``purity,'' and $\su3$-QDS, which they labelled ``coherence,'' in
the ground band. Their diagnostic was the algebraic $\sigma$-fluctuation $\Delta\sigma_{\rm gs}$
together with the intrinsic-state formalism; they noted in passing that this fluctuation ``has the
same physical content as wave-function entropy,'' but did not construct a reduced density matrix,
a block purity, or a bipartite entanglement entropy. The $\Delta\sigma_{\rm gs}$ is itself a
quantum-number fluctuation, and $\su3$ decompositions of IBM eigenstates are standard; what is new
here is not that a symmetry label can fluctuate but the state-resolved, spectrum-wide comparison
that follows, and the demonstration that this fluctuation and the bipartite entanglement are
structurally distinct and not monotonically related in this family. We reproduce their setting and render both words as literal density-matrix
quantities. Our $\oo6$ machinery matches their $\sigma$-fluctuation $\Delta\sigma_{\rm gs}=2.473$
at the $\uu5$ limit (their value $2.47$) and gives an exactly sharp $\oo6$ point. On their own
states we then compute the $\oo6$ block purity $\Tr\rho_\sigma^2$ and the bipartite $s|d$
entanglement, and find them to be \emph{structurally distinct}, non-monotonically related axes: the
ground state becomes $\oo6$-pure ($\Tr\rho_\sigma^2 = 0.81$--$0.88$) exactly where its $s|d$
entanglement is \emph{maximal} (Fig.~\ref{fig:kremer}). This concretely realizes the logical
distinction asserted after Proposition~\ref{prop:equiv}: a state can carry an almost exact
representation label (high block purity) while being maximally entangled across the boson bipartition,
the two being functions of different structures that need not track one another. The QDS
``purity/coherence'' distinction is thus real and quantitative, but it is a statement about sector
structure, not about bipartite entanglement magnitude.
\section{Physical anchor: $^{168}$Er}
\label{sec:er168}

\begin{figure}[t]
\centering
\includegraphics[width=\columnwidth]{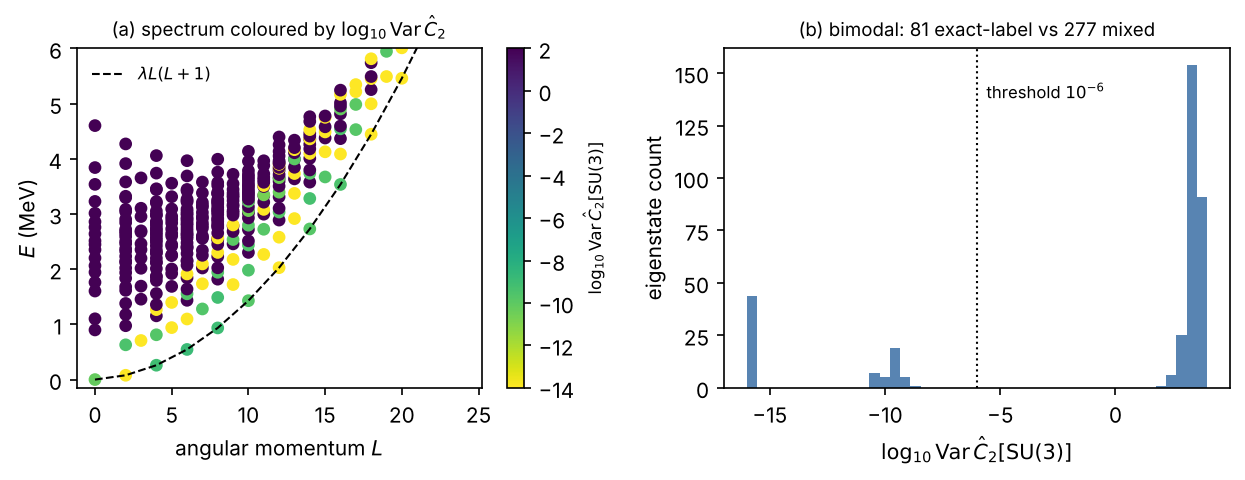}
\caption{Physical anchor at Leviatan's fitted $^{168}$Er parameters ($h_0=0.008$, $h_2=0.004$,
$\lambda=0.013$~MeV; the quantitative $^{168}$Er results in the text are at $N=16$, this panel uses
$N=12$ for legibility). (a) The $\su3$-PDS spectrum coloured by
$\log_{10}\Var\,\C[\su3]$: the solvable ground band (dark, $\Var\approx10^{-14}$) sits exactly on
the rigid-rotor line $\lambda L(L+1)$ (dashed), embedded in the mixed background (light). (b) The
label variance is sharply bimodal, cleanly separating the exact-label states from the mixed
remainder.}
\label{fig:er168}
\end{figure}

To test the diagnostic at realistic parameters we anchor in $^{168}$Er, a well-deformed
rare-earth rotor. We use Leviatan's fitted values for this nucleus~\cite{Leviatan1996},
$h_0=0.008$, $h_2=0.004$~MeV with the rotational parameter $C=\lambda=0.013$~MeV, at $N=16$; these
are the parameters of his Fig.~1, extracted from the $2^+_g$, $2^+_\gamma$ and $0^+_\beta$ energies.
The solvable ground band follows a rigid rotor, $E(2^+)=0.078$, $E(4^+)=0.260$, $E(6^+)=0.546$~MeV,
matching the diagonalized eigenvalues to $\lesssim10^{-15}$~MeV (the $(2N,0)$ band energy is
$N$-independent and exact), and the $\gamma$ band head sits at $E(2^+_\gamma)=6h_2(2N-1)+6\lambda
=0.822$~MeV. The label variance is bimodal and rigid at these physical parameters, with the
solvable entanglement frozen to $9\times10^{-13}$ under the $h_0$ flow; the exact-label count
follows the pattern established above, $81, 108, 139$ at $N=12,14,16$, of which $79, 106, 137$ are
closed-form solvable and the remaining two per case are the multiplicity partners of
Sec.~\ref{sec:crit}.

What Fig.~\ref{fig:er168} shows is not that the rotor is recovered, that is built into the
Hamiltonian, but that the label variance partitions the spectrum \emph{from the wavefunctions
alone}, with no reference to energies or transition rates, and the partition coincides with the
band assignment a spectroscopist would make from the level scheme. The diagnostic reconstructs the
rotational-band membership structurally rather than by fitting.

The block probabilities $P_\lambda$ are not only a classifier; on the mixed states they are the
$\su3$ admixtures that Leviatan quotes, and they connect to measured transition strengths. At the
fitted parameters the $\beta$-band head ($0^+$, $E=1.22$~MeV) has
$P_{(28,2)}=0.87$, $P_{(26,0)}=0.10$, $P_{(24,4)}=0.03$, reproducing Leviatan's stated $10\%$
$(26,0)$ and $3\%$ $(24,4)$ admixtures into the dominant $(28,2)$ irrep, while the ground and
$\gamma$ bands are $P_{(32,0)}=1$ and $P_{(28,2)}=1$ to numerical precision. Because the ground and
$\gamma$ bands are exactly irrep-pure and the $\su3$ quadrupole cannot connect different $(\lambda,
\mu)$, the $\gamma\!\to\! g$ $B(E2)$ ratios computed from these wave functions are parameter-free in the
sense that they are independent of the IBM \emph{Hamiltonian} parameters; they do still depend on the
structure of the $E2$ operator and on its fitted ratio $\theta/\alpha$ below.
Evaluating the general one-body $E2$ operator $T(E2)=\alpha\hat Q + \theta(d^\dagger s + s^\dagger
\tilde d)$ with Leviatan's fitted $\theta/\alpha=4.261$ on the diagonalized states reproduces his
$B(E2)$ table: the $2^+_\gamma\!\to\!(0^+_g,2^+_g,4^+_g)$ ratios come out $64.3:100:6.3$, matching
his PDS column to the quoted digits. Of these quantities the rotational scale $C=\lambda$ is the one Hamiltonian input, fixed to the observed $E(2^+)\simeq0.08$~MeV; the higher band energies, the $\su3$ admixtures, and the $B(E2)$ ratios are then outputs of the diagonalization, so the agreement is a reproduction of Leviatan's PDS results from the wavefunctions rather than a separate fit. Within the SU(3)-PDS description, then, the vanishing label
variance of these states explains why the corresponding inter-band ratios are parameter-free: it is the
exact irrep purity of the ground and $\gamma$ bands that makes those ratios independent of the
Hamiltonian parameters. This is a consistency statement within the model, not an independent
experimental validation of the diagnostic. A large label variance, conversely, flags the $\su3$-mixed states whose
inter-band $B(E2)$ strengths depend on the admixtures and can deviate from the ideal rotational
pattern. A full quantitative map from label variance to $B(E2)$ deviations across the rare-earth
region is left for future work; here we establish that the structural side of that correspondence
holds at a fitted, spectroscopically constrained parameter set.

\section{Relation to other quantum-information resources}
\label{sec:complementary}

The label variance does not coincide with the two other QI resources we examined at the stable PDS
point. The quantum Fisher information (QFI) for the collective quadrupole generator $\hat Q_0$
(QFI $=4\,\Var\,\hat Q_0$), a witness of multipartite entanglement, gives the exact-label states a
\emph{higher} QFI density than the mixed states, and the gap widens with system size:
$5.2$ versus $3.9$ ($N=6$), $6.3$ versus $4.5$ ($N=8$), $7.3$ versus $5.2$ ($N=10$). This
corroborates the fixed-$L$ result that symmetry-pure states are more, not less, entangled, but the
QFI does not reproduce the sharp zero-versus-large separation of the label variance and orders
pure above mixed for only $57\%$ of within-$L$ pairs; it serves as a supporting rather than a
defining measure. Nonstabilizerness (magic), quantified by the exact stabilizer $2$-R\'enyi
entropy $M_2$ under the dense qubit encoding, shows no statistically significant correlation with the
label variance in the tested data (Pearson $r=0.10$ over the $18$ states of the $N=4$ block, with a
$95\%$ interval of about $[-0.39,0.54]$ that spans zero; a comparable value at $N=6$): the exact-label set reaches lower magic than the
mixed states while remaining broadly distributed, and magic varies \emph{within} a single symmetry
band, so it resolves structure the label variance does not and conversely. Both statements are
encoding- and generator-dependent: QFI depends on the chosen generator and $M_2$ on the qubit
encoding, and the correlation is measured over a finite state sample, so we treat the qualitative
separations, not the specific values, as the robust conclusions, and we phrase them as
distinctness in the cases studied rather than as mathematical independence. The message aligns with
the concern that bipartite entropy misses multipartite structure~\cite{Brokemeier2025}: PDS
organizes several facets of complexity at once, with the solvable subset simultaneously carrying
zero label variance, cone rigidity, substantial entanglement, and access to near-stabilizer states.

\section{Relation to prior work}
\label{sec:prior}

The closest prior art is the sustained interacting-boson entanglement programme of the Tabriz
group~\cite{Jafarizadeh2022,Ghapanvari2025}, which computes the $s|d$ von Neumann entropy of the
ground-state coherent (mean-field) condensate and tracks its variation across the Casten triangle,
finding it minimal at $\uu5$ and maximal at $\oo6$. That work is orthogonal to ours along several
axes, summarized in Table~\ref{tab:distinction}: state studied (ground-state condensate
versus full exact spectrum), quantity (entropy magnitude versus label variance), role of
entanglement (order parameter that varies versus a magnitude that is explicitly non-diagnostic),
symmetry content (dynamical-symmetry limits versus partial dynamical symmetry), parameter dependence
(entropy tracks the control parameter versus rigidity as invariance along the PDS-preserving cone),
and criticality (a single mean-field entropy peak versus two distinct label-variance fingerprints by
transition order). That programme does not address partial dynamical symmetry, and the label
variance, block coherence, block purity, and cone rigidity are neither computed nor required there.

We are also not simply restating known algebra in quantum-information words. The $\su3$
decomposition of IBM eigenstates and the $\sigma$-fluctuation of Kremer \emph{et al.} are
established tools; what the present framework adds is a single state-resolved number, computed from
two expectation values without an explicit irrep decomposition, applied uniformly across an entire
spectrum, that (i) separates the exact-label subset from the mixed background where the entanglement
magnitude does not, (ii) distinguishes irrep purity from closed-form solvability through the
companion energy check, (iii) yields two distinct critical-point fingerprints set by the order of
the transition, and (iv) is defined identically on prepared states, so it transfers to a device
readout (Sec.~\ref{sec:qsim}). Symmetry-sector probabilities, quantum-number fluctuations, and
coherence relative to a symmetry decomposition are of course long established across quantum many-body
physics; what we claim is narrower. To our knowledge the state-resolved use of quadratic-Casimir
variance to identify the symmetry-selected subset of an IBM partial-dynamical-symmetry spectrum,
together with its comparison to bipartite entanglement and its device-oriented readout, has not been
systematically explored. Quantum-number fluctuations, the purity of symmetry components, embedded
and quasi-dynamical representations, and entropies of algebraic decompositions have each appeared in
nuclear structure; what is new here is their combination into a single spectrum-wide, state-resolved
PDS fingerprint with the entanglement contrast and the readout.

\begin{table}[t]
\caption{Distinction from the closest prior interacting-boson entanglement work.}
\label{tab:distinction}
\begin{ruledtabular}
\begin{tabular}{p{0.20\columnwidth} p{0.34\columnwidth} p{0.34\columnwidth}}
Axis & Prior IBM-entanglement work & This work \\
\colrule
State studied & ground-state condensate (mean field) & full exact spectrum, state by state \\
Quantity & $s|d$ entropy & label variance / block coherence \\
Role of entanglement & order parameter that varies & magnitude is non-diagnostic \\
Symmetry content & dynamical symmetry only & partial dynamical symmetry \\
Parameter dependence & entropy tracks control parameter & rigidity $=$ invariance on the cone \\
Criticality & single mean-field peak & two fingerprints by order \\
\end{tabular}
\end{ruledtabular}
\end{table}

\subsection{Scope and limitations}
The results are established within the $sd$ interacting boson model, for the specific PDS
Hamiltonians of Leviatan~\cite{Leviatan1996,Leviatan2007,Leviatan2020} and the examples we
diagonalized, at the boson numbers stated in each case. We do not claim the label variance as a
universal detector of partial dynamical symmetry across arbitrary many-body systems; what we show is
that, in this model and these examples, it separates the exact-label states from the mixed ones
where the magnitude of entanglement does not, and that the equivalence of
Proposition~\ref{prop:equiv} holds for any subalgebra with a quadratic Casimir, with the
irrep-label reading requiring the injectivity proviso discussed there (satisfied by the $\oo6$,
$\oo5$ and $\uu5$ labels used here, and enforced for $\su3$ by the joint $\C$--$\Cthree$ check). The
quoted separations ($\Var\,\C=0$ versus large mixed-state values), state counts, thresholds, and correlations are
numerical and example-specific, computed in the $M=0$ block with the degeneracy resolution described
in Sec.~\ref{sec:measurable}; the algebraic band-projection counts and the closed-form band-head
energies are exact. Extending the analysis to other algebras, to the neutron--proton (IBM-2) and
Bose--Fermi variants, and to a quantitative link with measured $B(E2)$ systematics is left for
future work.
\section{Quantum simulation of the fingerprint}
\label{sec:qsim}

The results above are obtained by exact diagonalization. Two features of the diagnostic make its
behavior on prepared states a question in its own right, rather than a routine simulation exercise.
First, it is a state-by-state quantity, and partial dynamical symmetry concerns a subset of states
scattered through the spectrum, not the ground state alone, so any device evaluation must prepare
excited states selectively, the setting where most quantum-information-in-nuclei work stops at the
ground state. Second, the label variance is defined on individual eigenstates, but a variational
solver returns, within a degenerate energy level, an arbitrary rotation of that level's states; we
show below that this obstructs a direct measurement of the label variance and that a specific
post-processing step, a Casimir diagonalization restricted to a degenerate energy window, removes
it. This is the concrete requirement a hardware implementation of the diagnostic must meet, and it
is the same degeneracy resolution used in the exact analysis (Sec.~\ref{sec:measurable}), now
carried onto prepared states.

We therefore encode the model on a qubit register, verify the encoding against the exact spectrum,
prepare the symmetry-pure and mixed eigenstates variationally, evaluate the label variance on them,
and show what the degeneracy resolution buys; a Trotterized propagator supplies the primitive a
rigidity-based test would use. Throughout, a statevector simulation of a model that is classically
diagonalizable is \emph{not} superior to classical diagonalization: the variational route is a
state-preparation and readout demonstration, and time evolution enters only as a tool, not as a
claim of quantum advantage. We follow the digital-simulation approach established for related
nuclear collective models~\cite{PerezFernandez2022,GarciaRamos2024}.

\subsection{Encoding and its verification}

The $sd$-IBM has six single-boson modes ($s$ and the five $d$ substates) and, crucially, no
truncation problem: the total boson number $N$ is the linear Casimir of $\uu6$, so the Hilbert
space is finite and its dimension is $\binom{N+5}{5}$. This is easier than generic bosonic
encodings, whose modes carry an infinite-dimensional Hilbert space requiring a controlled cutoff.
Standard schemes trade qubit count against operator locality: a one-hot (unary) register uses
$6(N+1)$ qubits with sparse, local ladder operators; a binary register per mode uses
$6\lceil\log_2(N+1)\rceil$ qubits with ladder operators that become sums of Pauli strings; and a
dense embedding of the $N$-conserving space uses $\lceil\log_2\binom{N+5}{5}\rceil$ qubits, with
compressed Gray-code variants having a nuclear-physics precedent~\cite{Singh2025}.
Table~\ref{tab:qubits} lists the counts; the compressed schemes place realistic $N$ within a
$16$--$30$ qubit register. We note that the qubit count is only part of the resource picture: the
dense encoding lowers the register size at the price of Casimir witnesses $\C[G]$ and $\C[G]^2$ that
become sums of many Pauli strings of higher weight, so a full hardware accounting must also track
the Pauli-term count and the Pauli weight. We return to this in Sec.~\ref{sec:resource}, where the
cost is found to grow, not fall, as one moves from the dense register to the locality-preserving
ones.

\begin{table}[t]
\caption{Qubit-count estimates for the $N$-conserving $sd$-IBM (six modes). ``binary/mode'' uses
$6\lceil\log_2(N+1)\rceil$ qubits; ``dense'' uses $\lceil\log_2\binom{N+5}{5}\rceil$.}
\label{tab:qubits}
\begin{ruledtabular}
\begin{tabular}{rrrrr}
$N$ & Hilbert dim & unary & binary/mode & dense \\
\colrule
2  & 21    & 18  & 12 & 5 \\
3  & 56    & 24  & 12 & 6 \\
6  & 462   & 42  & 18 & 9 \\
10 & 3003  & 66  & 24 & 12 \\
16 & 20349 & 102 & 30 & 15 \\
\end{tabular}
\end{ruledtabular}
\end{table}

We verify that the mapped Hamiltonian reproduces the exact spectrum before any state preparation.
Embedding the physical Hamiltonian into $n=\lceil\log_2 d\rceil$ qubits (with the unphysical padding
levels shifted above the spectrum) and decomposing into Pauli strings, the eigenvalues of the qubit
operator restricted to the physical block reproduce the exact spectrum to
$\max|E_{\rm qubit}-E_{\rm exact}| = 1.1\times10^{-14}$ ($N=2$, $5$ qubits, $216$ Pauli terms) and
$2.8\times10^{-14}$ ($N=3$, $6$ qubits, $1551$ terms). The encoding is exact.

\begin{figure}[t]
\centering
\includegraphics[width=\columnwidth]{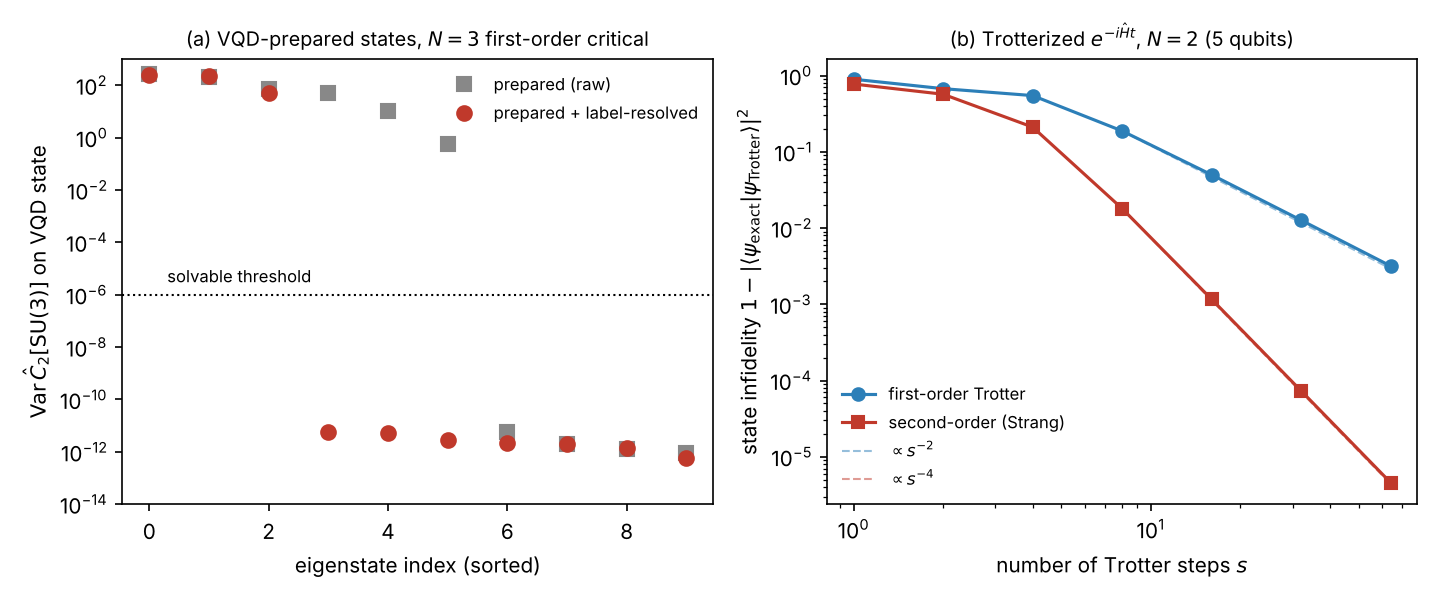}
\caption{Quantum simulation of the fingerprint (statevector, Aer). (a) The label variance
$\Var\,\C[\su3]$ evaluated on VQD-prepared eigenstates of the first-order critical Hamiltonian
($N=3$, $4$ qubits). Evaluated naively (grey squares), degenerate manifolds are returned as
arbitrary superpositions that mix labels and blur the separation; resolving the label within each
degenerate energy window (a post-processing step) recovers the sharp pure/mixed split (red
circles). (b) Trotter error of $e^{-i\hat Ht}$ against the number of steps $s$ for a physical
initial state ($N=2$, $5$ qubits, $t=0.5$): first-order and second-order (Strang) products follow
the expected $s^{-2}$ and $s^{-4}$ infidelity scaling (dashed guides).}
\label{fig:qsim}
\end{figure}

\subsection{Variational state preparation and the label variance}

We prepare eigenstates with the variational quantum eigensolver (VQE) for the ground state and
variational quantum deflation (VQD) for the excited states~\cite{Higgott2019}, VQD being essential
since the solvable states of a PDS Hamiltonian are scattered through the spectrum. We use the
hardware-efficient real ansatz of Fig.~\ref{fig:ansatz}: an initial layer of single-qubit $R_y$
rotations, followed by $L$ repetitions of a ring-of-CNOTs entangling block (a linear chain closed by
a wrap-around CNOT) and a further $R_y$ rotation layer, giving $n(L+1)$ real parameters on $n$
qubits. The rotations are kept real ($R_y$ only) because the Hamiltonian is real-symmetric in the
Fock basis, so its eigenvectors can be chosen real and no complex phases are needed; this halves the
parameter count relative to a generic ansatz. For the $N=3$ block ($n=4$) we take $L=6$, giving $28$
parameters. Optimization is gradient-based (L-BFGS-B) with multi-restart warm starts on the Aer
statevector simulator~\cite{Qiskit}; the deflation penalty $\beta\sum_{k}|\langle\phi_k|\psi\rangle|^2$
with $\beta$ large orthogonalizes each new state against those already found. We prepare all ten
eigenstates of the $N=3$ first-order critical block (a $10$-dimensional $M=0$ space) with
$\max|\Delta E| = 3.6\times10^{-13}$, far below the $10^{-4}$ target we require before trusting any
derived quantity.

\begin{figure}[t]
\centering
\includegraphics[width=\columnwidth]{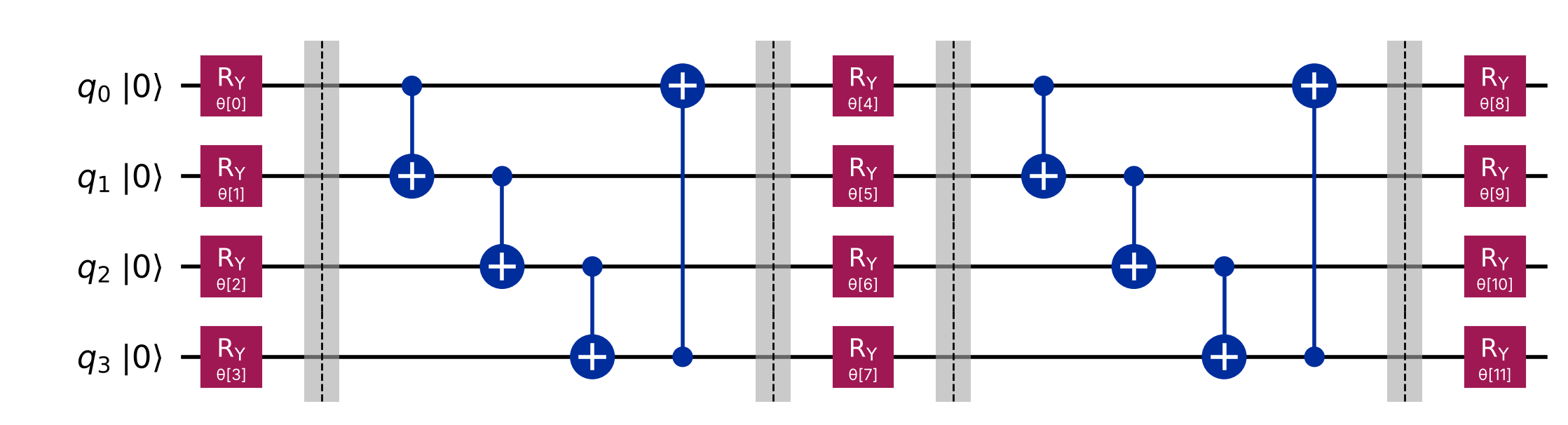}
\caption{The hardware-efficient real ansatz used for VQE/VQD state preparation, shown for $n=4$
qubits (the $N=3$ block) with $L=2$ entangling blocks for legibility; the runs use $L=6$. Each
block is a ring of CNOTs (linear chain plus a wrap-around gate) followed by a layer of parametrized
$R_y(\theta)$ rotations, preceded by an initial $R_y$ layer, for a total of $n(L+1)$ real parameters.
Only $R_y$ rotations appear because the Fock-basis Hamiltonian is real-symmetric and its eigenvectors
can be chosen real. Barriers mark the block boundaries and are not gates.}
\label{fig:ansatz}
\end{figure}

Reading the label variance off the prepared states is not automatic, and the reason is the crux of
what a device implementation must handle. A variational solver returns, within a degenerate energy
level, an \emph{arbitrary} superposition of that level's exact eigenstates; such a superposition
mixes irrep labels and so shows a spuriously large $\Var\,\C[\su3]$ even when the level is spanned by
symmetry-pure states [Fig.~\ref{fig:qsim}(a), grey]. The label variance is well defined on
individual eigenstates, but degeneracy, which is generic in these symmetric spectra and is exactly
where the solvable bands live, makes the naive prepared-state estimate unreliable. The fix is the
same degeneracy resolution used in the exact analysis (Sec.~\ref{sec:measurable}), now applied to
prepared states: within each degenerate energy window one forms the Casimir matrix on the prepared
states of that window and diagonalizes it, which recovers the sharp separation
[Fig.~\ref{fig:qsim}(a), red], seven of the ten states registering as symmetry-pure after
resolution, matching the exact count up to a single edge state at this small $N$, where a
$\uu5$-pure state shares a level with $\su3$-pure states. The step is modest in circuit terms, but
it is not free: diagonalizing $\C[G]$ on a $g$-fold degenerate window requires the off-diagonal
matrix elements $\bra{\psi_i}\C[G]\ket{\psi_j}$ across the window, i.e.\ an interference (or
subspace) measurement among the prepared states, not a single expectation value. It needs no deeper
state-preparation circuit and no additional ansatz coherence, so the diagnostic remains preparable
and measurable on hardware provided degeneracies are grouped and the within-window Casimir matrix is
estimated before the label variance is read.

Three practical points qualify this on a real device, and we state them because the degeneracy step is
the crux of a hardware implementation. First, the window is exact here, grouped with tolerance $\delta
E=10^{-9}$; on hardware the window must be widened to the energy-estimation uncertainty, which merges
true degeneracies with near-degeneracies, so the method generalizes to an \emph{approximate}-window
grouping whose width is set by the achievable energy resolution rather than by a sharp threshold.
Second, diagonalizing the Casimir on the window assumes the prepared states span that degenerate
manifold; if VQD misses a state or returns nearly linearly dependent approximations, the recovered
labels are incomplete, so the deflation must be run to the point where the window is fully covered and
the prepared states are checked for linear independence before the Casimir matrix is formed. Third, the
$|\Delta E|\sim10^{-13}$ we quote is a statevector benchmark and bounds only the ideal case; it says
nothing about variational robustness under gate noise, which is the separate question addressed, for
the readout rather than the preparation, in Sec.~\ref{sec:devicenoise}.

\subsection{Trotterized time evolution}

Time evolution enters the project only as a tool, and only here. We construct the Trotterized
propagator $e^{-i\hat Ht}$ for the encoded Hamiltonian and benchmark it against exact matrix
exponentiation for a physical initial state [Fig.~\ref{fig:qsim}(b), $N=2$, $5$ qubits, $t=0.5$].
The first-order product formula converges with state infidelity $\propto s^{-2}$ in the number of
steps $s$, and the second-order (Strang) formula with $\propto s^{-4}$, reaching an infidelity of
$4.6\times10^{-6}$ at $s=64$. The encoded model is thus straightforwardly time-simulable,
and this supplies the primitive that a rigidity-based test would use: evolve under two members of a
PDS-preserving operator family and threshold the drift of the prepared state's entanglement, a
protocol that requires no fit to spectroscopic data. None of the exact-diagonalization results of
Secs.~\ref{sec:stable}--\ref{sec:prior} depends on time evolution. What the propagator buys is a
dynamical reading of the diagnostic itself, which we turn to now.

\subsection{The dynamical form: a Casimir echo}
\label{sec:echo}

The same propagator turns the static diagnostic into a dynamical one, and this is where real-time
evolution stops being a benchmark and becomes the measurement. Evolving a prepared state under the
Casimir and recording the return probability gives the Loschmidt echo
\begin{equation}
E(t)=\bigl|\bra{\psi}e^{-i\C[G]t}\ket{\psi}\bigr|^2
     =\Bigl|\textstyle\sum_\lambda P_\lambda(\psi)\,e^{-if_2(\lambda)t}\Bigr|^2,
\label{eq:echo}
\end{equation}
where $P_\lambda(\psi)$ is the weight of $\ket{\psi}$ in the irrep block $\lambda$ and $f_2(\lambda)$
the Casimir eigenvalue. A symmetry-pure state carries a single label, is an eigenstate of $\C[G]$,
and stays rigid, $E(t)=1$ for all $t$; a mixed state superposes blocks with distinct $f_2(\lambda)$
and dephases, so $E(t)<1$ at generic times. Because the Casimir spectrum is integer-valued in this
normalization, the echo is strictly periodic and returns to unity at revival times; the dephasing is
its behavior between revivals. The short-time behavior ties the two faces together,
\begin{equation}
E(t)=1-t^2\,\Var_\psi\C[G]+O(t^4),
\label{eq:echoshort}
\end{equation}
the remainder being $O(t^4)$ because $E(t)$ is a squared modulus, hence even in $t$, so no odd-power
term appears.
so the echo is the label variance read as a bounded survival probability: the unitary performs,
through phase interference, the separation that the static estimator performs by subtracting two
moments of order $N^4$. Figure~\ref{fig:echo} shows both faces for the $N=3$ first-order critical
block. The symmetry-pure states sit at $E=1$ to machine precision for every $t$, the mixed states
fall to $E\simeq0.66$, and at short times the mixed-state echo follows $1-t^2\,\Var\,\C[\su3]$ to six
digits. This is the form a device test of the rigidity would read. The echo is a dynamical
representation of the same sector distribution the static variance measures, not a separate result;
the static variance remains the primary quantity.

\begin{figure}[t]
\centering
\includegraphics[width=\columnwidth]{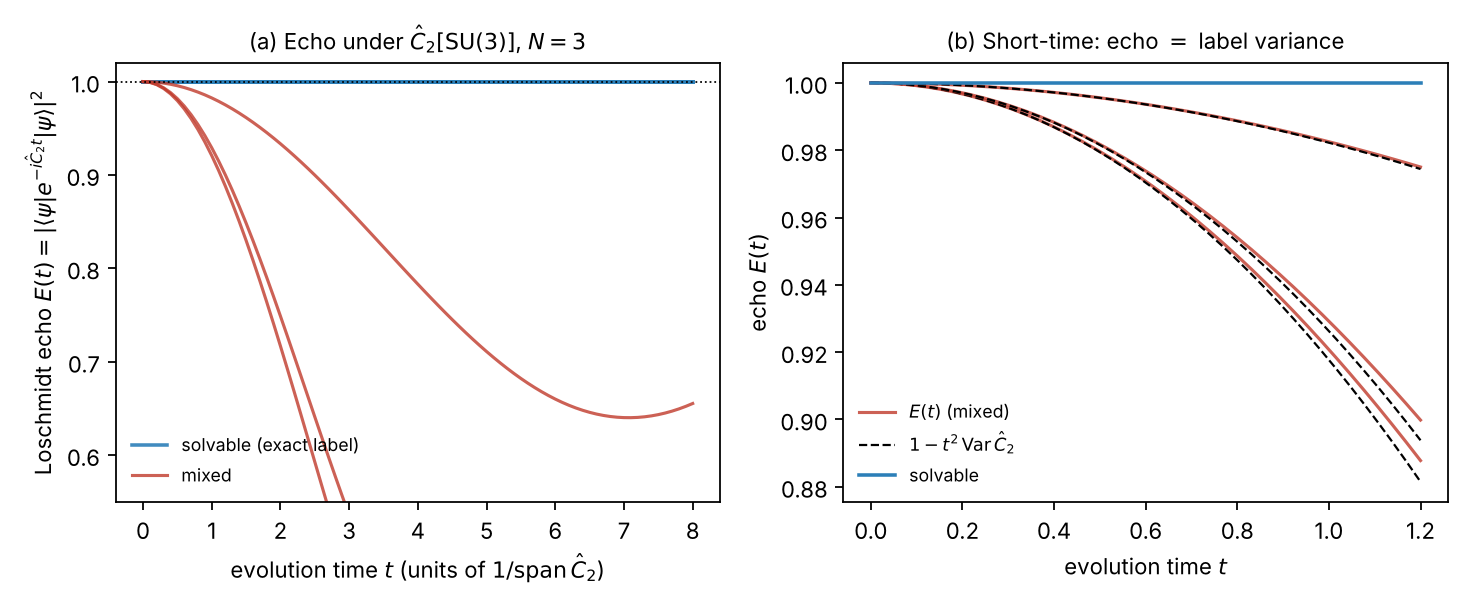}
\caption{The dynamical form of the diagnostic ($N=3$ first-order critical block). (a) Loschmidt echo
$E(t)$ under $\C[\su3]$: the symmetry-pure states are rigid at $E=1$ for all $t$ while the mixed
states dephase. (b) At short times the mixed-state echo follows $1-t^2\,\Var\,\C[\su3]$ (dashed),
confirming that the echo is the label variance expressed as a bounded quantity.}
\label{fig:echo}
\end{figure}

That a dynamical reading exists does not mean a cheap one does. One might hope to replace the
degeneracy-resolved Casimir measurement by a single low-weight operator that already labels the
symmetry-pure subspace. In the compressed encoding no such operator exists: scanning every one- and
two-qubit Pauli, the largest magnitude any of them attains on the symmetry-pure set is
$|\avg{P}|=0.27$, far from the value $1$ a label would need, and wherever a Pauli is nearly sharp on
the pure states it is equally sharp on the mixed ones. The symmetry-pure subspace is fixed by
collective $\su3$ labels, which have no few-qubit representation once the space is compressed.
Measuring the diagnostic is therefore a genuine task on a register, not a change of notation, and
the same collective structure that makes it a real probe makes its dynamical evaluation costly to
compile (Sec.~\ref{sec:resource}).

\subsection{Readout under device noise}
\label{sec:devicenoise}

The static readout is a shallow circuit, a state preparation of about fifteen two-qubit gates after
transpilation followed by a Casimir measurement, but a fragile observable, since on a symmetry-pure
state the variance is the near-cancellation of two moments of order $N^4$ that noise disrupts. The
two properties pull against each other, and the shallow depth is what lets mitigation win.
Figure~\ref{fig:noise} reports a density-matrix simulation of the $N=3$ readout under a heavy-hex
noise model: the pure-to-mixed contrast $\Var_{\rm mix}/\Var_{\rm pure}$ stays resolvable, between
two and three, up to two-qubit error rates of a few parts in $10^3$, provided the exact conservation
of the boson number $N$ and the projection $M$ is used for symmetry post-selection, which discards
leakage at no circuit cost. Without that post-selection the contrast is marginal. The noise model is a
heavy-hex depolarizing model with two-qubit error swept over $10^{-3}$ to $10^{-2}$, single-qubit error
$3\times10^{-4}$, and readout error $10^{-2}$; it contains stochastic (depolarizing and readout) errors,
not coherent ones, and the contrast is computed from the exact noisy density matrix, so the figure
carries no shot noise and the curve is the noise-limited ceiling rather than a finite-sample estimate.
A ready-to-run script for the readout on a superconducting processor is included with the code; on this
model, zero-noise extrapolation alone leaves a reduced but non-trivial contrast, and the symmetry
post-selection is the decisive ingredient. The aim is to show that the diagnostic is numerically
resolvable under the tested noise model, not that it outruns classical diagonalization, which for
this model it cannot; a direct run on \texttt{ibm\_kingston} is reported next.

We then ran the readout on IBM's $156$-qubit \texttt{ibm\_kingston} processor through Qiskit
Runtime\footnote{Qiskit Runtime job ID \texttt{db2u5su8v0ts73c3i75g}, 7 October 2026; qiskit 2.5.2,
qiskit-ibm-runtime 0.50.0. The state preparations used physical qubits $\{73,79,72,74\}$ (solvable,
transpiled depth $61$) and $\{148,60,47,21\}$ (mixed, depth $0$: in the dense encoding this exemplar is
the all-zero computational-basis state $\ket{0}^{\otimes4}$ and so needs no preparation gates, the
measured variance coming from the basis-change and readout circuits of the Casimir measurement alone).
Full metadata and the raw output are
in the repository.} (\textsc{EstimatorV2}, $4000$ shots, resilience level~2: Twirled Readout Error
eXtinction (TREX) and zero-noise extrapolation, and \emph{without} symmetry post-selection). The
reported variance is the plug-in estimator $\widehat{\Var}\,\C=\widehat{\avg{\C^2}}-\widehat{\avg{\C}}^2$,
formed from the two separately estimated moments. On the mixed state, prepared by a shallow circuit, it
is $\Var\,\C[\su3]=236\pm78$, matching the exact value $240$ in point estimate; this is individually
inconsistent with zero at about $3\sigma$. On the solvable state, prepared by the full $\sim\!15$
two-qubit-gate circuit, it is $36\pm136$, consistent with its exact value of zero
(Fig.~\ref{fig:kingston}). The point estimates therefore match the exact classification, the device
reproducing the large mixed variance to $1.5\%$ and leaving the solvable variance at the noise floor.
The uncertainties are large because the readout is a difference of two moments of order $N^4$ estimated
at only $4000$ shots, precisely the shot-budget sensitivity anticipated above; the quoted errors are
moreover conservative, propagated as if the two estimated moments were independent, which omits the
positive $\C$--$\C^2$ correlation that reduces the variance of their difference. The point-estimate
ratio of the two variances is $6.6$, but the direct difference between the mixed and solvable estimates
is resolved only at about $1.3\sigma$ under this conservative independent-error estimate, so a
statistically decisive two-class discrimination would require more shots and the covariance term above,
left to a dedicated hardware study. The run is thus a pattern-consistency demonstration, current
hardware reproducing the predicted large-variance versus near-zero-variance pattern at a single $N=3$
exemplar pair, rather than a decisive hardware classification. The raw output, with per-moment standard
errors, is included with the code.

\begin{figure}[t]
\centering
\includegraphics[width=0.78\columnwidth]{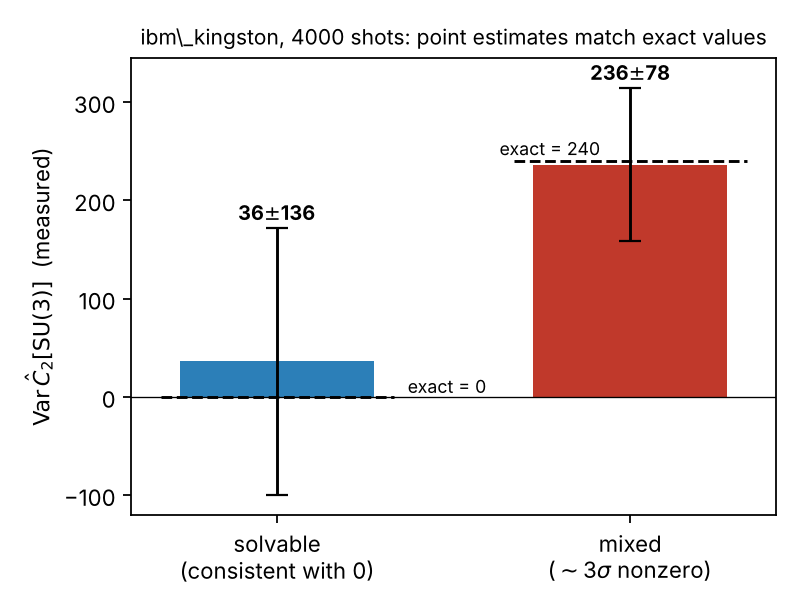}
\caption{Label-variance readout on the \texttt{ibm\_kingston} superconducting processor ($N=3$,
$4000$ shots, TREX $+$ zero-noise extrapolation, no symmetry post-selection; error bars are
conservative, propagating the per-moment standard errors as if the two moments were independent). The
mixed state reproduces its exact value of $240$ ($236\pm78$, individually nonzero at $\sim\!3\sigma$);
the solvable state ($36\pm136$) is consistent with its exact value of $0$. The point estimates match the
exact classification; the direct mixed$-$solvable difference is resolved only at $\sim\!1.3\sigma$ at
this shot count, so the run is a pattern-consistency demonstration, not a decisive two-class
discrimination.
Qiskit Runtime job ID \texttt{db2u5su8v0ts73c3i75g}.}
\label{fig:kingston}
\end{figure}

Two practical costs deserve emphasis. First, the observable is a difference of moments, so with
$\delta(\Var\,\C)\simeq[(\delta\langle\C^2\rangle)^2+4\langle\C\rangle^2(\delta\langle\C\rangle)^2
-4\langle\C\rangle\,\mathrm{Cov}(\C,\C^2)]^{1/2}$ and $\langle\C^2\rangle\sim N^4$, a pure state, whose
variance is the near-cancellation of two large moments, needs the more shots the larger $N$ is; we
display the covariance term because $\C$ and $\C^2$ are estimated from overlapping measurements and are
partially correlated. Second, $\C$ and $\C^2$ must be estimated from their Pauli decompositions: in the
dense $N=3$ encoding $\C[\su3]$ has $128$ Pauli terms falling into $39$ qubit-wise-commuting
measurement groups, and $\C^2$ has $136$ terms in $41$ groups, so a single label-variance reading uses
of order $80$ measurement settings for this $N=3$ dense encoding under qubit-wise-commuting grouping.
This figure is not a universal cost: it depends on the grouping strategy, on whether $\C^2$ is
decomposed independently or estimated through correlated estimators sharing $\C$'s settings, and on
whether symmetry post-selection is applied, and it grows with the register as the compressed encoding
raises the Pauli weight (Sec.~\ref{sec:resource}). A full shot-budget and measurement-grouping analysis
is left to a dedicated hardware study; here the section is a proof of principle, not a resource study.

\begin{figure}[t]
\centering
\includegraphics[width=\columnwidth]{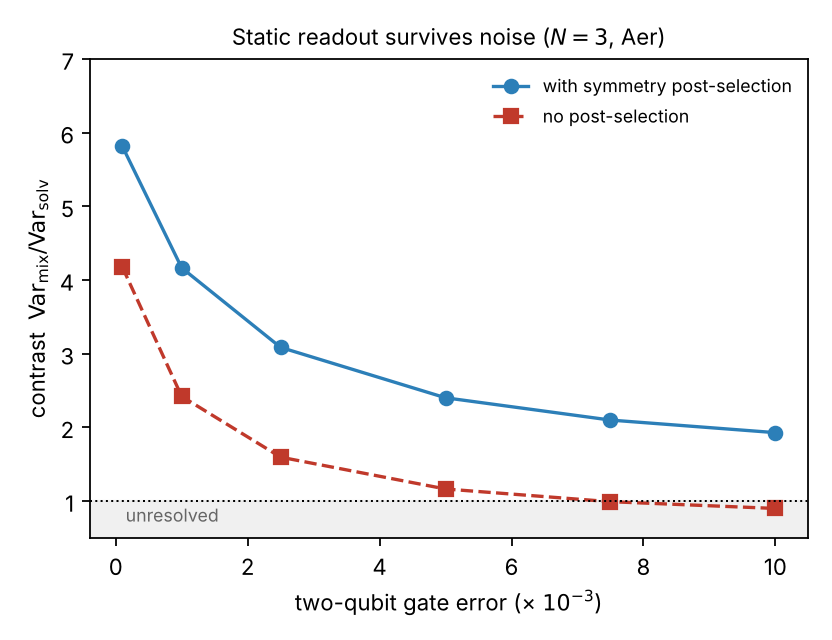}
\caption{Static label-variance readout under a heavy-hex device noise model ($N=3$, density-matrix
simulation). The pure-to-mixed contrast stays resolvable across the tested two-qubit error rates with
symmetry post-selection (exact $N$ and $M$ conservation) and is marginal without it; the shallow
preparation circuit keeps the fragile observable within reach of mitigation.}
\label{fig:noise}
\end{figure}

\subsection{Resource trade-off}
\label{sec:resource}

The dynamical form, by contrast, is not yet within device reach at useful size, and the reason is
structural rather than technological. Opening an $O(1)$ echo gap means evolving under $\C[G]$ for a
time over which, in the compressed encoding, the propagator compiles to a deep circuit, since
$\C[G]$ there is a non-local operator of high Pauli weight. A locality-preserving boson-to-qubit
encoding does not remove this under a naive Pauli-Trotterization: moving from the dense register to
a binary-per-mode and then a one-hot register increases both the number of Pauli terms in the
Hamiltonian and their maximum weight, because the physical operator must now be written over a much
larger unphysical space. At $N=3$ the dense encoding carries $1568$ Pauli terms of maximum weight
$6$, the binary-per-mode encoding about $1.6\times10^4$ terms of weight $8$, and the one-hot encoding
(at $N=2$) more than $2\times10^5$ terms of weight $12$. The collective content that makes the
diagnostic a genuine probe is the same content that makes its time evolution expensive to compile.
Evaluating the dynamical form at scale therefore calls for a structured, non-Pauli implementation of
the boson ladder operators, which we leave to future work, and the regime in which a quantum
evaluation would be preferable to classical diagonalization lies beyond the $sd$-IBM, in the larger
model spaces where the Hilbert space is no longer small. Within the studied $sd$-IBM settings, the
present work establishes the diagnostic, its dynamical form, and its measurability; the advantage
question belongs to that larger setting.

\section{Summary and outlook}
\label{sec:outlook}

We have shown that partial dynamical symmetry leaves a sharp signature in the eigenstates of the
interacting boson model that is independent of spectroscopic fitting: the vanishing of a
symmetry-label fluctuation on a structurally selected set of states. This fluctuation has two
companions that share its zero, block-coherence entropy and block impurity, and the block-purity
face makes the ``purity'' and ``coherence'' language of the quasi-dynamical-symmetry literature
literal. The label variance is numerically zero on the exact-label subset and large on the
mixed states; it certifies irrep purity, which the accompanying closed-form energy check separates
from the stronger property of solvability; it persists at first- and second-order critical points
as two different signatures set by the order of the transition; it holds at Leviatan's fitted
$^{168}$Er parameters, where the block probabilities reproduce his $\su3$ admixtures and the
parameter-free $\gamma\!\to\! g$ $B(E2)$ ratios; and it carries information distinct from
multipartite entanglement and from magic in the cases studied. The magnitude of bipartite
entanglement, by contrast, does not separate the two families and drifts even where the labels are
exact. The separation does not close with boson number (to $N=16$), degrades only as $\epsilon^2$
under a controlled breaking of the symmetry, and does not appear on generic non-PDS Hamiltonians, so
within a family already of PDS type the variance isolates the symmetry-selected states while a
separate energy check establishes solvability. We state the result at that scope: a state-resolved
structural fingerprint of the symmetry-selected eigenstates of the $sd$-IBM, complementary to
entanglement, not a universal detector of partial dynamical symmetry.

As a state-resolved criterion that does not use the energy spectrum or transition strengths in the
classification once the Hamiltonian eigenstates are obtained, the diagnostic can help classify fitted
IBM Hamiltonians and screen for candidate PDS structures within an independently fitted IBM description. It also has an exact dynamical form, a
Loschmidt echo under Casimir evolution that is rigid on the symmetry-pure states
(Sec.~\ref{sec:echo}), and a readout we ran on \texttt{ibm\_kingston}, where the measured variances
reproduce the predicted pattern, a mixed value of $236\pm78$ against the exact $240$ and a solvable
value consistent with zero (Sec.~\ref{sec:devicenoise}). Two tasks remain before a device evaluation becomes
compelling rather than illustrative. The first is a structured, non-Pauli compilation of the boson
ladder operators, without which the dynamical form stays out of reach at useful size
(Sec.~\ref{sec:resource}). The second is to move beyond the $sd$-IBM, whose Hilbert space is small
enough that classical diagonalization is never bettered, into the larger model spaces, the
neutron-proton and Bose-Fermi variants and configuration mixing, where a quantum evaluation could
earn its place. The broader result is a quantum-information reading of partial dynamical symmetry, a
notion developed so far with mean-field and group-theoretic tools.

\begin{acknowledgments}
The author thanks Dr.\ Subhendu Rajbanshi of the Department of Physics, Presidency University, for
the use of his laboratory as a Visiting Researcher and for continuous support, valuable discussions,
and guidance throughout this work.
\end{acknowledgments}

\section*{Data availability}
The code and numerical data that support the findings of this study are openly available in the
public repository of Ref.~\cite{repo}. The repository contains the $sd$-IBM basis and operator construction, the
exact-diagonalization and symmetry-sector routines (including the $\C$--$\Cthree$ check), the
critical-point and Kremer reproductions, the $^{168}$Er anchor with the $B(E2)$ evaluation, the
variational state-preparation and Trotter scripts, the Casimir-echo and device-noise analyses, the
binary and one-hot encodings with their spectrum verification, the low-weight-label scan, the
system-size scaling, approximate-PDS perturbation, and non-PDS false-positive sweeps, a ready-to-run
superconducting-processor readout script, and the scripts that generate every figure and table. The
\texttt{ibm\_kingston} run (Qiskit Runtime job ID \texttt{db2u5su8v0ts73c3i75g}) is provided as the
retrieved \textsc{EstimatorV2} output: the per-moment expectation values $\avg{\C}$, $\avg{\C^2}$ and
their standard errors for both states, the backend name and version, the physical-qubit and
transpilation metadata, and the script that regenerates Fig.~\ref{fig:kingston}.

\end{document}